\documentclass[]{interact}

\usepackage{epstopdf}
\usepackage[caption=false]{subfig}
\usepackage{multirow}
\usepackage{mathrsfs}
\usepackage{xcolor}
\usepackage{textcomp}
\usepackage{manyfoot}
\usepackage{algorithm}
\usepackage{algorithmicx}
\usepackage{algpseudocode}
\usepackage{float}
\usepackage{listings}
\usepackage{url}
\usepackage{natbib}
\bibpunct[, ]{(}{)}{;}{a}{,}{,}

\usepackage[hidelinks]{hyperref}
\usepackage{newunicodechar}
\newunicodechar{—}{---}

\theoremstyle{plain}
\newtheorem{theorem}{Theorem}

\newtheorem{corollary}{Corollary}
\newtheorem{lemma}{Lemma}

\theoremstyle{definition}
\newtheorem{assumption}{Assumption}

\theoremstyle{remark}
\newtheorem{remark}{Remark}

\providecommand{\bmhead}[1]{\subsection*{#1}}

\begin{document}

\articletype{ARTICLE}

\title{Optimal subsampling approach for Cox models under large-scale length-biased right-censored data}

\author{
\name{Zhicheng Feng\textsuperscript{a}\thanks{Email: zhichengfeng0@gmail.com}, Panpan Ba\textsuperscript{a}\thanks{Email: bpp0218@163.com}, Jianhua Shi\textsuperscript{a,b,c}\thanks{CONTACT Jianhua Shi. Email: v0085@126.com}, and Lin Ouyang\textsuperscript{d,e}\thanks{Email: ddcqzg@126.com}}
\affil{\textsuperscript{a}School of Mathematics and Statistics, Minnan Normal University, Zhangzhou, Fujian 363000, China; \textsuperscript{b}Fujian Key Laboratory of Granular Computing and Applications, Zhangzhou 363000, China; \textsuperscript{c}Fujian Key Laboratory of Data Science and Statistics, Zhangzhou 363000, China; \textsuperscript{d}Department of Radiology, The 909th Hospital, Zhangzhou 363020, China; \textsuperscript{e}Institute of Medical Imaging, Medical College, Xiamen University, Xiamen 361000, China}
}

\maketitle

\begin{abstract}
With the rapid advancement of data acquisition technologies, massive survival data have become increasingly prevalent, posing substantial computational challenges for the standard Cox proportional hazards model under complex censoring and truncation mechanisms. This paper addresses the analysis of large-scale length-biased right-censored data within the Cox model framework. We propose a novel subsampling approach to enhance computational efficiency for large samples and establish theoretically that the resulting subsample estimator provides an efficient approximation to the full-data estimator, while preserving consistency and asymptotic normality. Explicit optimal subsampling probabilities are derived under A- and L-optimality criteria. To facilitate practical implementation, we develop an efficient two-step estimation algorithm. Extensive numerical simulations and an empirical analysis of U.S. airline flight operation data evaluate the estimation performance and computational gains under various censoring rates and sample sizes. The results demonstrate that the proposed method substantially reduces computational cost while preserving statistical accuracy, offering a practical and scalable solution for analyzing large-scale length-biased right-censored survival data.
\end{abstract}


\section{Introduction}\label{sec1}

Survival analysis provides a fundamental framework for modeling time-to-event outcomes and assessing associated risk factors, with broad applications across biomedicine, engineering, and economics. Among existing models, the Cox proportional hazards model \citep{cox1972} has emerged as a benchmark due to its semiparametric structure, which enables estimation of covariate effects without specifying the baseline hazard function. Its theoretical tractability and interpretability have made it a standard tool in modern survival analysis.

However, the rapid proliferation of large-scale biomedical data--including electronic health records, genomic profiles, and disease registries--has led to a dramatic increase in sample sizes and covariate dimensions. While such rich data offer unprecedented opportunities for discovery, they also pose substantial computational challenges for conventional full-data estimation methods. In particular, under massive sample sizes, limitations in memory and computational resources often render standard algorithms prohibitively time-consuming or even infeasible \citep{wang2018optimal,tarkhan2020bigsurvsgd}.

Meanwhile, in many observational studies and disease screening programs, survival data often exhibit a more complex structure known as length-biased right censoring. This phenomenon commonly arises when the sampling mechanism is inherently linked to disease duration, as occurs in cross-sectional surveys or prevalent cohort studies \citep{vardi1989multiplicative, wang1991nonparametric}. Unlike traditional incident cohorts, a prevalent cohort consists of individuals who have already experienced the initiating event (e.g., disease onset) but have not yet experienced the failure event (e.g., death) at the time of recruitment. Consequently, individuals with longer survival durations are more likely to be included in the sample than those with shorter durations, inducing a systematic bias between the observed sample distribution and the target population distribution. For recent related research on length-biased and right-censored data, one may refer to \cite{chen2015monotone}, \cite{shi2022nonparametric}, and \cite{qiu2022reweighting}. Ignoring this length bias and directly applying the standard Cox model leads to inconsistent estimation and inefficient statistical inference. To address this issue, bias correction methods have been developed. For example,  \cite{luo2009nonparametric} proposed an approach based on inverse probability weighting, while \cite{qin2010statistical} introduced a conditional likelihood method. However, these corrections typically increase the complexity of the likelihood function and the variance structure of the resulting estimators. In the context of massive data, a major challenge for contemporary statistical inference is to effectively correct length bias while maintaining computational feasibility and scalability for large sample sizes.

For overcoming the computational bottlenecks posed by massive data, two main classes of methods have been proposed: divide-and-conquer algorithms \citep{huang2019distributed,chen2022first} and subsampling algorithms. Among these, subsampling approaches construct estimators by selecting a small, information-rich subset from the full data, thereby substantially reducing computational cost while preserving estimation accuracy. This paradigm has attracted considerable attention in recent years. For example, \cite{ma2015} proposed a sketching-based subsampling method using leverage scores for large-scale linear regression. \cite{wang2018optimal} studied optimal subsampling for logistic regression under massive sample sizes. \cite{ai2021optimal} developed optimal subsampling methods for generalized linear models under the A-optimality criterion. In the field of survival analysis, \cite{zuo2021sampling} first extended optimal subsampling to the additive risk model. For the Cox proportional hazards model, \cite{keret2023analyzing} investigated optimal subsampling tailored for rare-event settings. Subsequently, \cite{qiao2024optimal} and \cite{zhang2024approximating} established more general optimal subsampling frameworks. More recently, \cite{Li2025} proposed a Poisson subsampling approach to better accommodate large-scale datasets. However, existing work on optimal subsampling for the Cox model has primarily targeted standard right-censored data and cannot be directly migrated to length-biased data. The reason is that the bias-corrected weighting structure fundamentally alters the asymptotic variance of the parameter estimator, rendering previously optimal sampling probabilities suboptimal--or even leading to efficiency loss--under length-biased right-censored data.

In this paper, we propose a novel optimal subsampling strategy tailored for large-scale length-biased right-censored data under the Cox proportional hazards model. The consistency and asymptotic normality of the resulting subsample estimator are established. Through theoretical proof, numerical simulations, and a real data application, we demonstrate the performance advantages of the proposed method and its practical utility.

The remainder of this paper is organized as follows. Section 2 presents the model specification for the length-biased right-censored data and describes the full-data estimation method. We then propose a general subsampling estimator, analyze its large-sample properties, derive the optimal subsampling probabilities, and detail the implementation steps of a two-step sampling algorithm. Numerical simulations evaluating the finite-sample performance of the proposed method are presented in Section 3. Section 4 illustrates the method using a real dataset. Section 5 concludes with a brief discussion. All proofs are provided in the Appendix.

\section{Data and models}\label{sec2}
\subsection{Length-biased right-censored data}
Let $\tilde{T}$ denote the true survival time from the initial event to the failure event, and $\tilde{A}$ the truncation time from the initial event to the sampling time. Under the length-biased sampling mechanism, an individual can be included in the study only if they have not experienced the failure event at the sampling time, i.e., $\tilde{T} \geq \tilde{A}$. To simplify the notation and emphasise the truncation mechanism, we drop the superscript $\tilde{}$ and use $A$ and $T$ to denote the truncation time and survival time of the observed sample, respectively. Then the joint distribution of $(A, T)$ is the same as the conditional distribution of $(\tilde{A}, \tilde{T})$ given $\tilde{T} \geq \tilde{A}$. Let $V = T - A$ be the residual survival time measured from the sampling time, and let $C$ be the residual censoring time. Accordingly, the actually observed time is $Y = A + \min(V, C)$, and the failure indicator is defined as $\delta = I(V \leq C)$. Finally, the observed dataset consists of $n$ independent and identically distributed (i.i.d.) samples, denoted as $\mathcal{D}_n = \{(A_i, Y_i, \delta_i, \boldsymbol{X}_i); i = 1, \dots, n\}$, where $\boldsymbol{X}_i$ is a $p$-dimensional covariate vector.

\begin{figure}[htbp]
	\centering
	\includegraphics[trim=0 4cm 0 0cm, clip, width=1
    \textwidth]{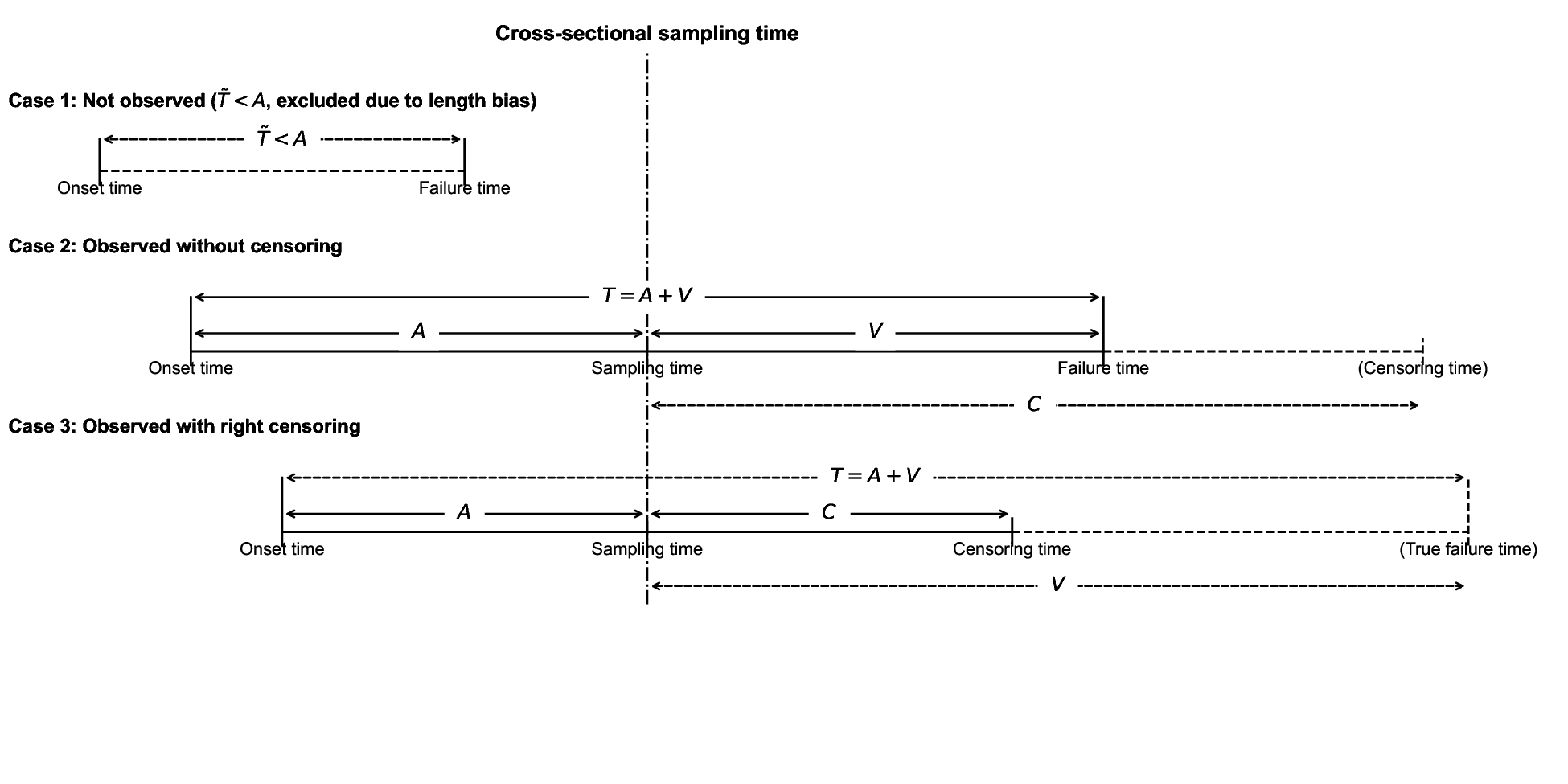}
	\caption{Schematic diagram of length-biased right-censored data.\label{f1}}
\end{figure}
Assuming that $Y_i$ is independent of $C_i$ given $\boldsymbol{X}_i$, and that the censoring process is independent of the covariates, the conditional hazard function at time $t$ for $i$-th individual can be formulated under the Cox proportional hazards model  as
\begin{eqnarray}\label{eq0}
    \lambda(t|\boldsymbol{X}_i)=\lambda_0(t)\exp\left(\boldsymbol{\beta}^T\boldsymbol{X}_i\right),
\end{eqnarray}
where $\lambda_0(t)$ denotes an unknown baseline hazard function, and $\boldsymbol{\beta}$ is a $p$-dimensional vector of regression parameters, which can be estimated by the  partial  likelihood  function of model (\ref{eq0}).

The classical Cox model assumes that the observed sample is obtained via simple random sampling from the target population. However, this assumption does not hold in the presence of length-biased right-censored data. This is because the probability of observation is proportional to the survival time, implying that individuals with longer survival durations are more likely to be included in the sample, thereby inducing selection bias. This bias leads to a distortion of the risk set and renders the conventional partial likelihood estimator inconsistent. To correct for this bias, a common approach is to construct weighted estimating equations. \cite{qin2010statistical} introduced two alternative weighting approaches. The EE-I method assigns weights of the form $\{Y_i S_c(Y_i - A_i)\}^{-1}$, whereas EE-II employs the inverse of the integrated survival function, given by $\{\int_0^{Y_i} S_c(t)\,dt\}^{-1}$. A key drawback of EE-I lies in its susceptibility to extreme weights as $S_c(\cdot)$ approaches zero. In contrast, the integral-based formulation of EE-II mitigates such instability in the tail region. As a consequence, EE-II yields more stable and efficient estimation of the regression parameters $\boldsymbol{\beta}$, particularly in large-scale survival settings. Based on the above considerations, we adopt the EE-II framework defined as
\begin{eqnarray}\label{eq1}
\mathbf{U}_{II}(\boldsymbol{\beta})
=\sum_{i=1}^{n}\delta_i\left[\mathbf{X}_i-\frac{\sum_{j=1}^{n}I(Y_j\ge Y_i)\delta_j W_j\mathbf{X}_j
\exp\left(\boldsymbol{\beta}^{\top}\mathbf{X}_j\right)}{
\sum_{j=1}^{n}I(Y_j\ge Y_i)\delta_j W_j
\exp\left(\boldsymbol{\beta}^{\top}\mathbf{X}_j\right)}\right]=0 ,
\end{eqnarray}
where $
W_j = \{w_c(Y_j)\}^{-1},$
and
$w_c(y) = \int_0^y S_c(t) \, dt,$
and $S_c(t)$ denotes the survival function of the censoring variable $C$.

Equivalently, the parameter $\boldsymbol{\beta}$ can be estimated by solving the following estimating equation
\begin{equation}\label{eq2}
    \mathbf{U}_{II}(\boldsymbol{\beta})= \sum_{i=1}^{n} \int_{0}^{\tau} \bigl\{ \mathbf{X}_{i} - \bar{\mathbf{X}}_{w}(t, \boldsymbol{\beta}) \bigr\} \, dN_{i}(t),
\end{equation}
where 
\begin{equation}
    \overline{\mathbf{X}}_w(t, \boldsymbol{\beta}) = \frac{S^{(1)}(t, \boldsymbol{\beta})}{S^{(0)}(t, \boldsymbol{\beta})},
\end{equation}
and
\begin{equation}
S^{(l)}(t, \boldsymbol{\beta}) = n^{-1} \sum_{i=1}^n w_c(t) R_i(t) W_i \mathbf{X}_i^{\otimes l} \exp(\boldsymbol{\beta}^\top \mathbf{X}_i),
 l = 0, 1 \text{ and } 2,
\end{equation}
with $R_i(t) = I({Y}_i \geq t)\delta_i, 
N_i(t) = I({Y}_i \leq t, C_i \geq Y_i - A_i). $
Here, $\boldsymbol{u}^{\otimes 0}=1$,$\boldsymbol{u}^{\otimes 1}=\boldsymbol{u}$, and $\boldsymbol{u}^{\otimes 2}=\boldsymbol{u}\boldsymbol{u}^\top$ for any vector $\boldsymbol{u}$.

\subsection{Asymptotic properties and subsampling}
In massive data analysis, the sample size frequently far exceeds available computational resources, rendering conventional full-sample estimation methods infeasible for direct implementation. The objective of this study is to construct a subsampling-based estimator that substantially reduces computational cost while providing an efficient approximation to the full-data estimator. Specifically, given prespecified sampling probabilities $\{\pi_i: i=1, \dots, n\}$ of one objective population, a random subsample of size $r$ is firstly drawn with replacement from the full data sample. The resulting subsample is denoted by $\mathcal{D}_{n}^{*}=\{(A_{i}^{*}, Y_{i}^{*}, \delta_{i}^{*}, \mathbf{X}_{i}^{*}, \pi_{i}^{*}); i=1, \dots, r\}$, where $\pi_{i}^{*}$ denotes the associated subsampling probability. Based on the subsampled data, we incorporate inverse probability weighting within the EE-II framework to construct a weighted estimating equation $\mathbf{U}^*(\boldsymbol{\beta})=0$, 
where
\begin{equation}\label{eq3}
 \mathbf{U}^*(\boldsymbol{\beta})
= \frac{1}{nr} \sum_{i=1}^{r} \frac{1}{\pi_i^*} \int_{0}^{\tau} \bigl\{ \mathbf{X}_i^* - \overline{\mathbf{X}}_w^*(t, \boldsymbol{\beta}) \bigr\} \, dN_i^*(t),
\end{equation}
with
\begin{equation}
\overline{\mathbf{X}}_w^*(t, \boldsymbol{\beta})
= \frac{S^{*(1)}(t, \boldsymbol{\beta})}{S^{*(0)}(t, \boldsymbol{\beta})},
\end{equation}
and
\begin{equation}
S^{*(l)}(t, \boldsymbol{\beta})
= \frac{1}{nr} \sum_{i=1}^{r} \frac{1}{\pi_i^*} w_c^*(t) R_i^*(t) W_i^* (\mathbf{X}_i^*)^{\otimes l} \exp(\boldsymbol{\beta}^\top \mathbf{X}_i^*), l = 0,1  \text{ and }  2,
\end{equation}
where
$
R_i^*(t) = I(Y_i^* \geq t)\delta_i^*,
N_i^*(t) = I(Y_i^* \leq t,\, C_i^* \geq Y_i^* - A_i^*),
W_i^* = \{w_c^*(Y_i^*)\}^{-1},
$
and $w_c^*(y) = \int_0^y S_c^*(t)\,dt.$ The subsampling estimator $\hat{\boldsymbol{\beta}}_{\text{sub}}$ is obtained by solving the weighted estimating equation $\mathbf{U}^*(\boldsymbol{\beta}) = 0$.

The asymptotic properties of the subsampling estimator $\hat{\boldsymbol{\beta}}_{sub}$ are established under the following conditions.
\begin{assumption}\label{ass1}
The baseline hazard function satisfies $\int_0^\tau \lambda_0(t)\,dt < \infty$ and $P(T_i \geq \tau) > 0$, where $\tau > 0$ is a given constant.
\end{assumption}

\begin{assumption}\label{ass2}
The true parameter $\boldsymbol{\beta}_0$ is an interior point of $\Theta$, a compact parameter space. Moreover, for any $\boldsymbol{\beta} \in \Theta$,
\[
\frac{1}{n}\sum_{i=1}^n \int_0^\tau \left[ \frac{S^{(2)}(t,\boldsymbol{\beta})}{S^{(0)}(t,\boldsymbol{\beta})}
- \left\{ \frac{S^{(1)}(t,\boldsymbol{\beta})}{S^{(0)}(t,\boldsymbol{\beta})} \right\}^{\otimes 2} \right] dN_i(t)
\]
converges in probability to a positive definite matrix.
\end{assumption}

\begin{assumption}\label{ass3}
The covariate vectors $\mathbf{X}_i$ are bounded.
\end{assumption}

\begin{assumption}\label{ass4}
The subsampling probabilities $\pi_i's$ satisfy $\max_{1 \le i \le n}(n\pi_i)^{-1} = O_P(1)$.
\end{assumption}

Assumptions \ref{ass1}-\ref{ass3} are standard regularity conditions for the large-sample theory of the Cox model under the full data, ensuring desirable statistical properties. Assumption \ref{ass4} is specific to the subsampling framework and is imposed to control the instability of the weighted estimation. Collectively, these assumptions guarantee the consistency and asymptotic normality of the subsampling estimator $\hat{\boldsymbol{\beta}}_{sub}$.

\begin{theorem}\label{th1}
If Assumptions \ref{ass1}-\ref{ass4} hold, then as $r, n \to \infty$, we have
\begin{equation}\label{eq4}
   \sqrt{r} \bigl( \hat{\boldsymbol{\beta}}_{sub} - \hat{\boldsymbol{\beta}}_{full} \bigr)
    \xrightarrow{d} N(0, \Sigma),
\end{equation}
where $\xrightarrow{d}$ denotes convergence in distribution, $\Sigma = \Psi^{-1} \Gamma \Psi^{-1}$
with
\begin{equation}\label{eq5}
    \Psi=\frac{1}{n}\sum_{i=1}^{n}\int_{0}^{\tau}\left[\frac{S^{(2)}(t,\hat{\boldsymbol{\beta}}_{full})}{S^{(0)}(t,\hat{\boldsymbol{\beta}}_{full})}-\left\{\frac{S^{(1)}(t,\hat{\boldsymbol{\beta}}_{full})}{S^{(0)}(t,\hat{\boldsymbol{\beta}}_{full})}\right\}^{\otimes2}\right]dN_{i}(t),
\end{equation}
and
\begin{equation}\label{eq6}
\Gamma = \frac{1}{n^2} \sum_{i=1}^{n} \frac{1}{\pi_i} \left[ \int_{0}^{\tau} \bigl( \mathbf{X}_i - \overline{\mathbf{X}}_w(t, \hat{\boldsymbol{\beta}}_{full}) \bigr) dN_i(t) \right]^{\otimes 2}.
\end{equation}
\end{theorem}

Theorem \ref{th1} establishes the asymptotic normality of $\hat{\boldsymbol{\beta}}_{sub}$, whose asymptotic variance depends on the sampling probabilities $\pi_i$. To evaluate the performance of parameter estimation, we derive the optimal subsampling probabilities by minimizing the mean squared error (MSE) under the A-optimality criterion \citep{wang2018optimal}.
\begin{theorem}\label{th2}
    If the following sampling probabilities are adopted,
    \begin{equation}\label{eq7}
        \pi_i^{Aopt} = \frac{\left\| \Psi^{-1} \int_{0}^{\tau} \bigl\{ \mathbf{X}_i - \overline{\mathbf{X}}_w(t, \hat{\boldsymbol{\beta}}_{full}) \bigr\} \, dN_i(t) \right\|}
{\sum_{j=1}^{n} \left\| \Psi^{-1} \int_{0}^{\tau} \bigl\{ \mathbf{X}_j - \overline{\mathbf{X}}_w(t, \hat{\boldsymbol{\beta}}_{full}) \bigr\} \, dN_j(t) \right\|},
\quad i=1,2,\dots,n,
    \end{equation}
    then the asymptotic MSE of $\hat{\boldsymbol{\beta}}_{sub}$ attains its minimum.
\end{theorem}
Notably, the expression of $\pi_i^{\mathrm{Aopt}}$ involves the matrix inverse $\Psi^{-1}$, which can incur a prohibitive computational burden in large-scale settings. To alleviate this issue, we adopt the modified optimality criterion proposed by \cite{wang2018optimal}, under which the modified optimal sampling probabilities are derived by minimizing the MSE of $\Psi \hat{\boldsymbol{\beta}}_{sub}$. This strategy effectively bypasses the need for calculating $\Psi^{-1}$, thereby significantly enhancing computational efficiency.
\begin{theorem}\label{th3}
    If the following sampling probabilities are adopted
    \begin{equation}\label{eq8}
        \pi_i^{Lopt} = \frac{\left\|\int_{0}^{\tau} \bigl\{ \mathbf{X}_i - \overline{\mathbf{X}}_w(t, \hat{\boldsymbol{\beta}}_{full}) \bigr\} \, dN_i(t) \right\|}
{\sum_{j=1}^{n} \left\|\int_{0}^{\tau} \bigl\{ \mathbf{X}_j - \overline{\mathbf{X}}_w(t, \hat{\boldsymbol{\beta}}_{full}) \bigr\} \, dN_j(t) \right\|},
\quad i=1,2,\dots,n,
    \end{equation}
    then the asymptotic MSE of $\Psi \hat{\boldsymbol{\beta}}_{sub}$ attains its minimum.
\end{theorem}
The optimal subsampling probabilities in Theorems \ref{th2} and \ref{th3} both depend on the full-data estimator $\hat{\boldsymbol{\beta}}_{{full}}$. However, in practical applications, computing $\hat{\boldsymbol{\beta}}_{{full}}$ in advance is computationally prohibitive and contradicts the original motivation of subsampling techniques. To ensure algorithmic feasibility, a pilot estimator $\hat{\boldsymbol{\beta}}_{r_0}$ is commonly employed as a plug-in alternative to $\hat{\boldsymbol{\beta}}_{\mathrm{full}}$ for constructing the approximate optimal probabilities $\pi_i^{\mathrm{app}}$. Corresponding to different optimization criteria, the explicit forms of the resulting approximate probabilities are given as follows:
\begin{equation}
\pi_i^{\mathrm{app\text{-}Aopt}} =
\frac{\left\| \Psi^{-1} \int_0^\tau \bigl\{ \mathbf{X}_i - \overline{\mathbf{X}}_w^*(t, \hat{\boldsymbol{\beta}}_{r_0}) \bigr\} \, dN_i(t)\right\|}{\sum_{j=1}^n
\left\| \Psi^{-1} \int_0^\tau \bigl\{ \mathbf{X}_j - \overline{\mathbf{X}}_w^*(t, \hat{\boldsymbol{\beta}}_{r_0}) \bigr\} \, dN_j(t)\right\|},\quad i=1,2,\dots,n,
\end{equation}
and
\begin{equation}
    \pi_i^{\mathrm{app\text{-}Lopt}} = \frac{\left\|\int_0^\tau
\bigl\{ \mathbf{X}_i - \overline{\mathbf{X}}_w^*(t, \hat{\boldsymbol{\beta}}_{r_0}) \bigr\} \, dN_i(t)\right\|}{\sum_{j=1}^n \left\| \int_0^\tau
\bigl\{ \mathbf{X}_j - \overline{\mathbf{X}}_w^*(t, \hat{\boldsymbol{\beta}}_{r_0}) \bigr\} \, dN_j(t)\right\|},\quad i=1,2,\dots,n.
\end{equation}
\begin{remark}
Since $\bar{\mathbf{X}}_w(t, \hat{\boldsymbol{\beta}}_{r_0})$ depends on the full sample data, it entails a heavy computational burden in large-scale settings. To address this issue, we recommend replacing the full-sample term $\bar{\mathbf{X}}_w(t, \hat{\boldsymbol{\beta}}_{r_0})$ with its subsample counterpart $\bar{\mathbf{X}}_w^*(t, \hat{\boldsymbol{\beta}}_{r_0})$. This is a reasonable choice because it can be shown that     $\sup_{t\in[0,\tau],\boldsymbol{\beta}\in\Theta}\mid\overline{\mathbf{X}}_w^*(t,\boldsymbol{\beta})-\overline{\mathbf{X}}_w(t,\boldsymbol{\beta})\mid\xrightarrow{P|\mathcal{D}_n}0.$ This uniform consistency justifies the substitution at the consistency level and provides a computationally efficient approximation to the full-sample term.
\end{remark}

In practical computation, the approximate probabilities $\pi_i^{\mathrm{app}}$ for some observations tend to zero. When such observations are selected into the subsample, the corresponding weights are magnified and may dominate the estimator under extreme circumstances, leading to substantial variance inflation. To avoid this instability, we do not directly use the original sampling probabilities. Instead, following the idea of defensive importance sampling \citep{hesterberg1995weighted}, we mix the original sampling probabilities with the uniform distribution in a proportional manner. Specifically, for any given sampling probability $\pi_i$, we may define its robustified version as follows:
\begin{equation}
    \pi_i^{rob}(\lambda)=(1-\lambda)\pi_i+\frac{\lambda}{n},\quad\lambda\in(0,1).
\end{equation}
This construction guarantees a positive lower bound for $n\pi_i^{\mathrm{rob}}(\lambda)$, thereby avoiding extreme weights and enhancing the robustness of the subsampling estimation.

To facilitate practical application, the implementation details of the proposed robust two-step subsampling method are summarized in Algorithm 1.
\begin{algorithm}[H]
\caption{Robust Two-step Subsampling Algorithm}
\begin{algorithmic}[1]

\State \textbf{Input:} Full dataset \( \mathcal{D}_n = \{(A_i, Y_i, \delta_i, \boldsymbol{X}_i), i=1,\dots,n\} \), initial subsample size \( r_0 \).

\State \textbf{Step 1.} Perform simple random sampling from \( \mathcal{D}_n \) to obtain a subsample \( \mathcal{D}_{r_0} \) of size \( r_0 \). Based on this subsample, compute the initial estimator \( \hat{\boldsymbol{\beta}}_{r_0} \) as a candidate estimator of \( \hat{\boldsymbol{\beta}}_{{full}} \), and construct the approximate optimal subsampling probabilities
\[
\pi_i^{{app,rob}}(\lambda)=(1-\lambda)\pi_i^{{app}} + \frac{\lambda}{n}, \quad \lambda \in (0,1).
\]

\State \textbf{Step 2.} Perform sampling with replacement from \( \mathcal{D}_n \) using probabilities \( \{\pi_i^{{app,rob}}(\lambda)\}_{i=1}^n \) to obtain a subsample \( \mathcal{D}_r \) of size \( r \). Solve the following equation
\[
U^*(\boldsymbol{\beta})=\frac{1}{nr}\sum_{i=1}^r \frac{1}{\pi_i^{{app,rob*}}(\lambda)}\int_0^\tau \{\mathbf{X}_i^*-\bar{\mathbf{X}}_w^*(t,\boldsymbol{\beta})\}dN_i^*(t)=0.
\]

\State \textbf{Output:} Subsampling estimator \( \tilde{\boldsymbol{\beta}}_{sub} \).

\end{algorithmic}
\end{algorithm}
The following theorem gives the asymptotic normality of the estimator $\tilde{\boldsymbol{\beta}}_{sub}$ based on the mixture sampling probability.

\begin{theorem}\label{th4}
Under Assumptions \ref{ass1}-\ref{ass3} and $r=o(n)$, as $r_0\to\infty$, $r\to\infty$, and $n\to\infty$, given $\mathcal{D}_n$ and $\hat{\boldsymbol{\beta}}_{r_0}$, the two-step estimator $\tilde{\boldsymbol{\beta}}_{sub}$ from Algorithm 1 approximates the full-sample estimator at the rate $r^{-1/2}$. The approximation error satisfies the following asymptotic normal distribution:
\begin{equation}\label{eq9}
\sqrt{r} \bigl( \tilde{\boldsymbol{\beta}}_{sub} - \hat{\boldsymbol{\beta}}_{full} \bigr)
\xrightarrow{d} N(0, \tilde{\Sigma}),
\end{equation}
where $\tilde{\Sigma} = \Psi^{-1} \tilde{\Gamma} \Psi^{-1}$, with
\begin{equation}\label{eq10}
    \Psi=\frac{1}{n}\sum_{i=1}^{n}\int_{0}^{\tau}\left[\frac{S^{(2)}(t,\hat{\boldsymbol{\beta}}_{full})}{S^{(0)}(t,\hat{\boldsymbol{\beta}}_{full})}-\left\{\frac{S^{(1)}(t,\hat{\boldsymbol{\beta}}_{full})}{S^{(0)}(t,\hat{\boldsymbol{\beta}}_{full})}\right\}^{\otimes2}\right]dN_{i}(t),
\end{equation}
and
\begin{equation}\label{eq11}
\tilde{\Gamma} = \frac{1}{n^2} \sum_{i=1}^{n} \frac{1}{\pi_i^{app,rob}\left( \lambda  \right)} \left[ \int_{0}^{\tau} \bigl( \mathbf{X}_i - \overline{\mathbf{X}}_w(t, \hat{\boldsymbol{\beta}}_{full}) \bigr) dN_i(t) \right]^{\otimes 2}.
\end{equation}
\end{theorem}
\begin{remark}
    The estimator $\hat{\boldsymbol{\beta}}_{full}$ based on the full sample converges to the true parameter $\boldsymbol{\beta}_0$ at rate $n^{-1/2}$. Since $r=o(n)$, $\hat{\boldsymbol{\beta}}_{full}$ and $\tilde{\boldsymbol{\beta}}_{sub}$ are asymptotically substitutable by $\boldsymbol{\beta}_0$ without affecting the limiting distribution. Consequently, the asymptotic results in Theorem \ref{th4} can be directly used for statistical inference on the true parameter. Moreover, in large-scale data analysis, subsampling typically satisfies $r \ll n$, and inference methods based on asymptotic normality are not only theoretically justified but also practically reasonable and feasible.
\end{remark}
\begin{corollary}
\label{cor1}
\normalsize
    Under Assumptions \ref{ass1}-\ref{ass3}, if $r=o(n)$, then as $r_0 \to \infty$, $r \to \infty$, and $n \to \infty$, the two-step estimator $\tilde{\boldsymbol{\beta}}_{sub}$ obtained from Algorithm 1 converges to the true parameter $\boldsymbol{\beta}_0$ at the rate $r^{-1/2}$. Specifically, $\| \tilde{\boldsymbol{\beta}}_{sub}-\boldsymbol{\beta}_{0}\Vert =O_{P}(r^{-1/2})$. In addition, $\tilde{\boldsymbol{\beta}}_{sub}$ satisfies asymptotic normality, that is,
\begin{equation}\label{eq12}
\sqrt{r}\left(\tilde{\boldsymbol{\beta}}_{sub}-{\boldsymbol{\beta}}_0\right)\xrightarrow{d}N(0,\tilde{\Sigma}),
\end{equation}
where $\tilde{\Sigma}$ is the same as that given in Theorem \ref{th4}.
\end{corollary}

In order to perform inference for the true parameter using $\tilde{\boldsymbol{\beta}}_{sub}$, it is necessary to estimate the variance-covariance matrix of $\tilde{\boldsymbol{\beta}}_{sub}$. The most straightforward approach is to replace $\hat{\boldsymbol{\beta}}_{full}$ with $\tilde{\boldsymbol{\beta}}_{sub}$ in the expression for the asymptotic variance-covariance matrix $\tilde{\Sigma}$, thereby obtaining a valid estimator. Although this method is theoretically justified, it still requires access to the full dataset $\mathcal{D}_n$. When dealing with massive datasets, the associated computational cost can be substantial, potentially offsetting the advantages of subsampling. Consequently, it is desirable to develop a less computationally demanding alternative. To this end, we propose a subsample-based covariance matrix estimator, denoted by $\breve{\Sigma}$, for the two-step estimator $\tilde{\boldsymbol{\beta}}_{sub}$. While a full-data approach is possible, it entails heavy computational burdens. To enhance efficiency, the proposed $\breve{\Sigma}$ relies solely on the information from the subsample $\mathcal{D}_n^*$, which is given by
\begin{equation}
    \breve{\Sigma} = \breve{\Psi}^{-1} \breve{\Gamma} \breve{\Psi}^{-1},
\end{equation}
where
\begin{equation}
    \breve{\Psi} = \frac{1}{rn} \sum_{i=1}^{r} \frac{1}{\pi_i^{app,rob*}\left( \lambda \right)} \int_0^{\tau} \left[ \frac{S^{*(2)}(t, \tilde{\boldsymbol{\beta}}_{sub})}{S^{*(0)}(t, \tilde{\boldsymbol{\beta}}_{sub})}
- \left\{ \frac{S^{*(1)}(t, \tilde{\boldsymbol{\beta}}_{sub})}{S^{*(0)}(t, \tilde{\boldsymbol{\beta}}_{sub})} \right\}^{\otimes2} \right] dN_i^*(t),
\end{equation}
and
\begin{equation}
    \breve{\Gamma} = \frac{1}{r n^2} \sum_{i=1}^{r} \frac{1}{(\pi_i^{app,rob*}\left( \lambda  \right))^2}
\left[ \int_0^{\tau} \bigl\{ \mathbf{X}_i^* - \overline{\mathbf{X}}_w^*(t, \tilde{\boldsymbol{\beta}}_{sub}) \bigr\} \, dN_i^*(t) \right]^{\otimes 2}.
\end{equation}
\section{Simulations}\label{sec3}
This section evaluates the finite-sample performance of the proposed optimal subsampling method through numerical simulations. The unbiased failure times $\tilde{T}_i$ are generated independently and identically from a Cox proportional hazards model:
$$\lambda\left(t\mid\mathbf{X}_i\right)=\lambda_0\left(t\right)\exp\left(\boldsymbol{\beta}_0^\top\mathbf{X}_i\right).$$
The baseline hazard function is set to $\lambda_0(t) = 2t$ and the true parameter vector $\boldsymbol{\beta}_0 = (-2, -1, 0, 1, 2)^\top$. The covariate vector is $\mathbf{X}_i = (X_{i1}, X_{i2}, ..., X_{ip})^\top$, with $p=5$. Two scenarios are considered for the covariate distributions.
\begin{enumerate}
\item[(a)] $X_{1}\sim\text{Bernoulli}(0.5)$, $X_{2}$ and $X_{3}$ are uniformly distributed on $[-0.5, 0.5]$, and $X_{4}, X_{5} \sim N(0, 1)$.
\item[(b)] $X_{1}\sim\text{Bernoulli}(0.5)$, and the remaining continuous covariates $\mathbf{X}_{-1} = (X_{2}, \ldots, X_{p})^T$ are generated from an equally weighted mixture of two multivariate normal distributions: $0.5\cdot N(-\mathbf{1}, \boldsymbol{\Upsilon}) + 0.5\cdot N(\mathbf{1}, \boldsymbol{\Upsilon})$, where $\boldsymbol{\Upsilon}_{ij} = 0.5^{|i-j|}$.
\end{enumerate}
Censoring times $C_i$ are generated from a uniform distribution $U(0, c)$, where $c$ is chosen to achieve censoring rates of approximately $20\%$ or $50\%$. Left-truncation times $A_i$ are independently drawn from a uniform distribution on $(0, \tau_0)$, and only survival times satisfying $\tilde{T}_i \geq  A_i$ are retained. The value of $\tau_0$ is set to be larger than the upper bound of $T_i$ to ensure that the stationarity assumption is satisfied. A total of $n = 10^6$ observations are generated, and subsample sizes of $r = 400, 600, 800,$ and $1000$ are considered. Each simulation setting is replicated $1000$ times. The empirical MSE is computed as follows:
\begin{equation}\label{eq16}
    \operatorname{MSE}\left(\tilde{\boldsymbol{\beta}}_{sub}\right)=\frac{1}{1000} \sum_{i=1}^{1000}\left\|\tilde{\boldsymbol{\beta}}_{sub}^{(i)}-\hat{\boldsymbol{\beta}}_{full}\right\|^{2},
\end{equation}
where $\tilde{\boldsymbol{\beta}}_{sub}^{(i)}$ represents the result of the $i$-th simulation.

We further investigate the impact of the pilot sample size $r_0$ on the estimation performance. Tables \ref{table:1} and \ref{table:2} present the MSE of the L-optimal and A-optimal estimators for pilot sample sizes $r_0$ ranging from 300 to 500. The results indicate that the MSE decreases slowly as $r_0$ increases, suggesting that the two-step subsampling procedure is robust to the choice of the pilot sample size. Based on this finding, we recommend that, under simulation settings similar to those considered in this paper, the pilot sample size be chosen to ensure stability without incurring high computational cost. In higher-dimensional settings or scenarios with higher censoring rates, a more conservative approach is to increase the pilot sample size appropriately.
\begin{table}[!h]
\centering
\caption{MSE of the L-optimal subsampling estimator for different subsample sizes $r_0$}
\label{table:1}
\begin{tabular}{l *{9}{c}}
\toprule
 & & \multicolumn{4}{c}{CR = 20\%} & \multicolumn{4}{c}{CR = 50\%} \\
\cmidrule(lr){3-6} \cmidrule(lr){7-10}
Case & $r_0$ & \multicolumn{4}{c}{$r$} & \multicolumn{4}{c}{$r$} \\
\cmidrule(lr){3-6} \cmidrule(lr){7-10}
 &  & 400 & 600 & 800 & 1000 & 400 & 600 & 800 & 1000 \\
\midrule
\multirow{3}{*}{(a)}
 & 300 & 0.0933 & 0.0587 & 0.0460 & 0.0345 & 0.0943 & 0.0640 & 0.0467 & 0.0375 \\
 & 400 & 0.0922 & 0.0589 & 0.0445 & 0.0361 & 0.0957 & 0.0605 & 0.0473 & 0.0381 \\
 & 500 & 0.0950 & 0.0644 & 0.0439 & 0.0354 & 0.0967 & 0.0654 & 0.0461 & 0.0378 \\
\midrule
\multirow{3}{*}{(b)}
 & 300 & 0.0581 & 0.0390 & 0.0272 & 0.0233 & 0.0605 & 0.0413 & 0.0303 & 0.0230 \\
 & 400 & 0.0579 & 0.0358 & 0.0278 & 0.0238 & 0.0585 & 0.0386 & 0.0295 & 0.0231 \\
 & 500 & 0.0591 & 0.0397 & 0.0262 & 0.0214 & 0.0601 & 0.0368 & 0.0300 & 0.0232 \\
\bottomrule
\end{tabular}
\end{table}
\begin{table}[!h]
\centering
\caption{MSE of the A-optimal subsampling estimator for different subsample sizes $r_0$}
\label{table:2}
\begin{tabular}{l *{9}{c}}
\toprule
 & & \multicolumn{4}{c}{CR = 20\%} & \multicolumn{4}{c}{CR = 50\%} \\
\cmidrule(lr){3-6} \cmidrule(lr){7-10}
Case & $r_0$ & \multicolumn{4}{c}{$r$} & \multicolumn{4}{c}{$r$} \\
\cmidrule(lr){3-6} \cmidrule(lr){7-10}
 &  & 400 & 600 & 800 & 1000 & 400 & 600 & 800 & 1000 \\
\midrule
\multirow{3}{*}{(a)}
& 300 & 0.0743 & 0.0487 & 0.0389 & 0.0293 & 0.0839 & 0.0541 & 0.0410 & 0.0312 \\
 & 400 & 0.0773 & 0.0489 & 0.0384 & 0.0284 & 0.0774 & 0.0534 & 0.0381 & 0.0317 \\
 & 500 & 0.0757 & 0.0500 & 0.0357 & 0.0279 & 0.0770 & 0.0519 & 0.0392 & 0.0308 \\
\midrule
\multirow{3}{*}{(b)}
 & 300 & 0.0504 & 0.0301 & 0.0224 & 0.0184 & 0.0523 & 0.0329 & 0.0246 & 0.0195 \\
 & 400 & 0.0466 & 0.0311 & 0.0227 & 0.0173 & 0.0493 & 0.0320 & 0.0242 & 0.0209 \\
 & 500 & 0.0484 & 0.0322 & 0.0231 & 0.0185 & 0.0485 & 0.0341 & 0.0243 & 0.0194 \\
\bottomrule
\end{tabular}
\end{table}

In addition to the pilot sample size, the mixture proportion $\lambda$ plays a critical role in balancing estimation efficiency and numerical stability. We further examine the effect of $\lambda$ on the performance of the algorithm. Tables \ref{table:new1} and \ref{table:new2} show that when $\lambda=0$, the method reduces to using the theoretically optimal sampling probabilities.  In this scenario, the resulting MSE is slightly smaller than that for $\lambda=0.1$, which is expected because the simulated data follow the prespecified model and the pilot estimator is accurate. To enhance robustness for practical data that may contain outliers or exhibit complex distributions, and to prevent weight inflation caused by extremely small sampling probabilities, we set $\lambda=0.1$ in the subsequent simulations and empirical study.
\begin{table}[!h]
\centering
\caption{MSE of the A-optimal subsampling estimator under different mixing proportions}
\label{table:new1}
\begin{tabular}{l c *{8}{c}}
\toprule
 & & \multicolumn{4}{c}{CR = 20\%} & \multicolumn{4}{c}{CR = 50\%} \\
 \cmidrule(lr){3-6} \cmidrule(lr){7-10}
Case & $r$ & \multicolumn{4}{c}{$\lambda$} & \multicolumn{4}{c}{$\lambda$} \\
\cmidrule(lr){3-6} \cmidrule(lr){7-10}
 &  & 0 & 0.1 & 0.3 & 0.5 & 0 & 0.1 & 0.3 & 0.5 \\
\midrule
\multirow{4}{*}{(a)}
& 400 & 0.0762 & 0.0791 & 0.0836 & 0.0906 & 0.0851 & 0.0878 & 0.1016 & 0.1176 \\
& 600 & 0.0491 & 0.0507 & 0.0528 & 0.0491 & 0.0544 & 0.0562 & 0.0630 & 0.0718 \\
& 800 & 0.0361 & 0.0370 & 0.0392 & 0.0507 & 0.0395 & 0.0420 & 0.0481 & 0.0557 \\
& 1000 & 0.0281 & 0.0290 & 0.0313 & 0.0528 & 0.0319 & 0.0340 & 0.0381 & 0.0433 \\
\midrule
\multirow{4}{*}{(b)}
& 400 & 0.0483 & 0.0512 & 0.0529 & 0.0556 & 0.0505 & 0.0555 & 0.0612 & 0.0712 \\
& 600 & 0.0303 & 0.0321 & 0.0332 & 0.0356 & 0.0348 & 0.0390 & 0.0431 & 0.0486 \\
& 800 & 0.0225 & 0.0243 & 0.0257 & 0.0268 & 0.0256 & 0.0290 & 0.0310 & 0.0355 \\
& 1000 & 0.0185 & 0.0198 & 0.0205 & 0.0218 & 0.0203 & 0.0234 & 0.0244 & 0.0280 \\
\bottomrule
\end{tabular}
\end{table}
\begin{table}[!h]
\centering
\caption{MSE of the L-optimal subsampling estimator under different mixing proportions}
\label{table:new2}
\begin{tabular}{l c *{8}{c}}
\toprule
 & & \multicolumn{4}{c}{CR = 20\%} & \multicolumn{4}{c}{CR = 50\%} \\
 \cmidrule(lr){3-6} \cmidrule(lr){7-10}
Case & $r$ & \multicolumn{4}{c}{$\lambda$} & \multicolumn{4}{c}{$\lambda$} \\
\cmidrule(lr){3-6} \cmidrule(lr){7-10}
 &  & 0 & 0.1 & 0.3 & 0.5 & 0 & 0.1 & 0.3 & 0.5 \\
\midrule
\multirow{4}{*}{(a)}
& 400 & 0.0924 & 0.0931 & 0.0950 & 0.0964 & 0.1034 & 0.1068 & 0.1157 & 0.1302 \\
& 600 & 0.0610 & 0.0612 & 0.0629 & 0.0628 & 0.0637 & 0.0660 & 0.0725 & 0.0833 \\
& 800 & 0.0437 & 0.0441 & 0.0439 & 0.0447 & 0.0479 & 0.0504 & 0.0552 & 0.0606 \\
& 1000 & 0.0361 & 0.0364 & 0.0372 & 0.0377 & 0.0380 & 0.0398 & 0.0440 & 0.0479 \\
\midrule
\multirow{4}{*}{(b)}
& 400 & 0.0576 & 0.0594 & 0.0607 & 0.0632 & 0.0609 & 0.0683 & 0.0732 & 0.0812 \\
& 600 & 0.0374 & 0.0398 & 0.0395 & 0.0411 & 0.0376 & 0.0430 & 0.0451 & 0.0523 \\
& 800 & 0.0290 & 0.0302 & 0.0304 & 0.0312 & 0.0299 & 0.0338 & 0.0359 & 0.0403 \\
& 1000 & 0.0233 & 0.0245 & 0.0243 & 0.0252 & 0.0242 & 0.0281 & 0.0296 & 0.0333 \\
\bottomrule
\end{tabular}
\end{table}

As can be seen from Figure \ref{fig:1}, with fixed $r_0 = 300$ and $\lambda=0.1$, the MSE of all methods decreases as the subsample size $r$ increases. This finding aligns with the asymptotic results in Theorem \ref{th3}, indicating that a larger subsample size leads to improved estimation precision. Furthermore, there are notable differences in performance across the methods. Under all considered scenarios, the proposed A-optimality (A-opt) and L-optimality (L-opt)  methods substantially outperform the uniform sampling method. This suggests that, in massive data settings, the uniform method is inefficient at selecting informative observations, resulting in considerable efficiency loss. Among the methods compared, the proposed A-opt method yields the best performance, while the L-opt method performs slightly worse, though the difference is not substantial. In contrast, the uniform sampling method consistently exhibits the largest MSE, indicating that equal-probability sampling fails to effectively identify informative observations in massive data settings, resulting in substantial efficiency loss.
\begin{figure}[htbp]
	\centering
	\includegraphics[width=0.9\textwidth]{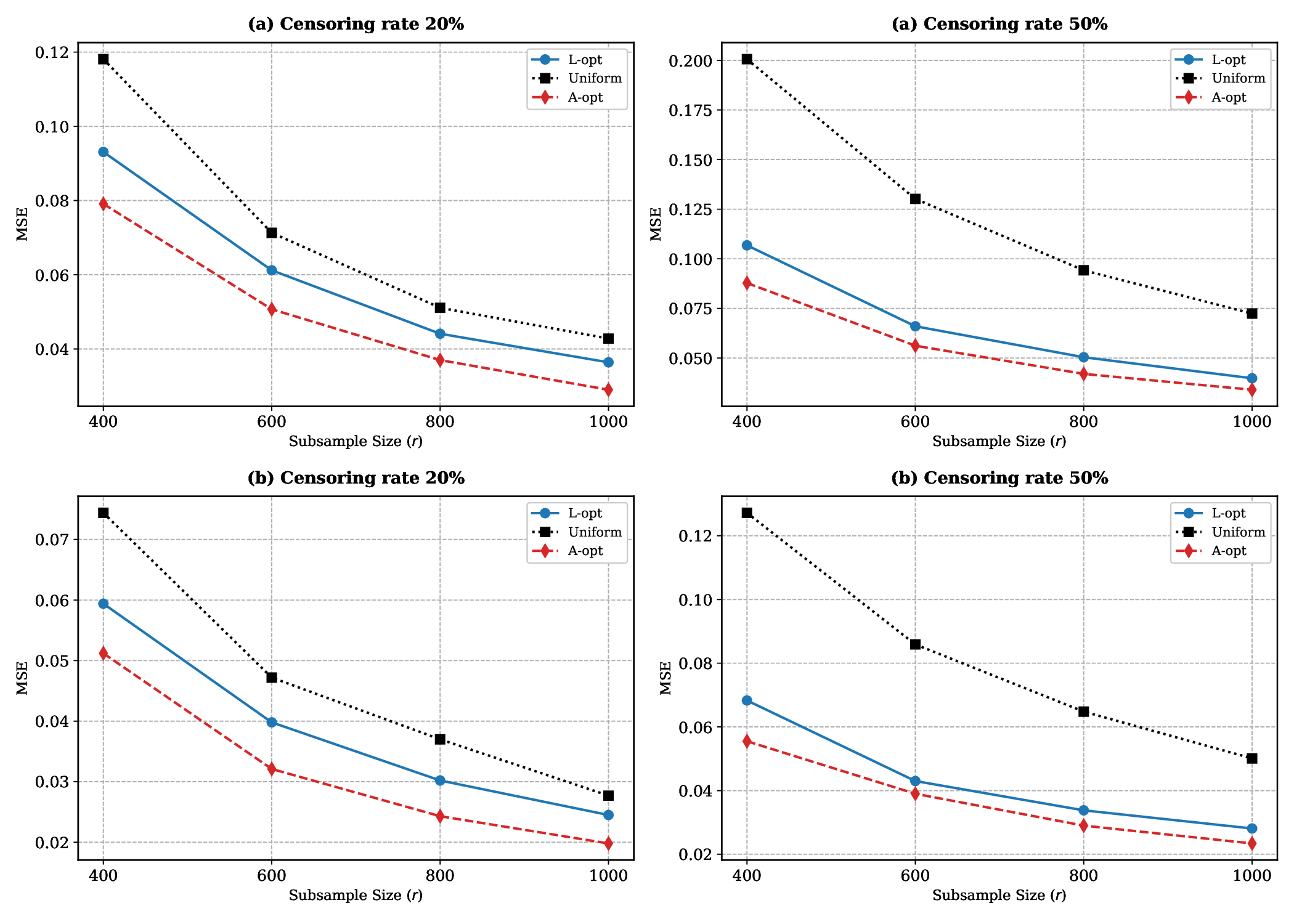}
	\caption{MSE of different subsampling methods.\label{fig:1}}
\end{figure}

To assess the performance improvement of the proposed optimal subsampling over simple uniform sampling, we define a relative efficiency measure as follows:
\begin{equation}\label{eq17}
    RE=\frac{\mathrm{MSE}_{\mathrm{unif}}}{\mathrm{MSE}_{\mathrm{opt}}}.
\end{equation}
Figure \ref{fig:2} presents the relative efficiency (RE) of different sampling methods. Under all settings, the relative efficiencies of both A-opt and L-opt exceed 1. Moreover, A-opt consistently outperforms L-opt in terms of estimation efficiency, indicating that the A-opt criterion captures more information. The censoring rate also has a substantial impact on the performance of the sampling strategies. When the censoring rate increases from 20\% to 50\%, the overall RE increases substantially, indicating that the advantage of optimal sampling is amplified under heavy censoring. Furthermore, as the subsample size $r$ increases from 400 to 1000, the relative efficiency remains relatively stable without a noticeable decline, suggesting that the proposed method maintains good stability across different sampling sizes.

\begin{figure}[!h]
	\centering
	\includegraphics[width=0.9\textwidth]{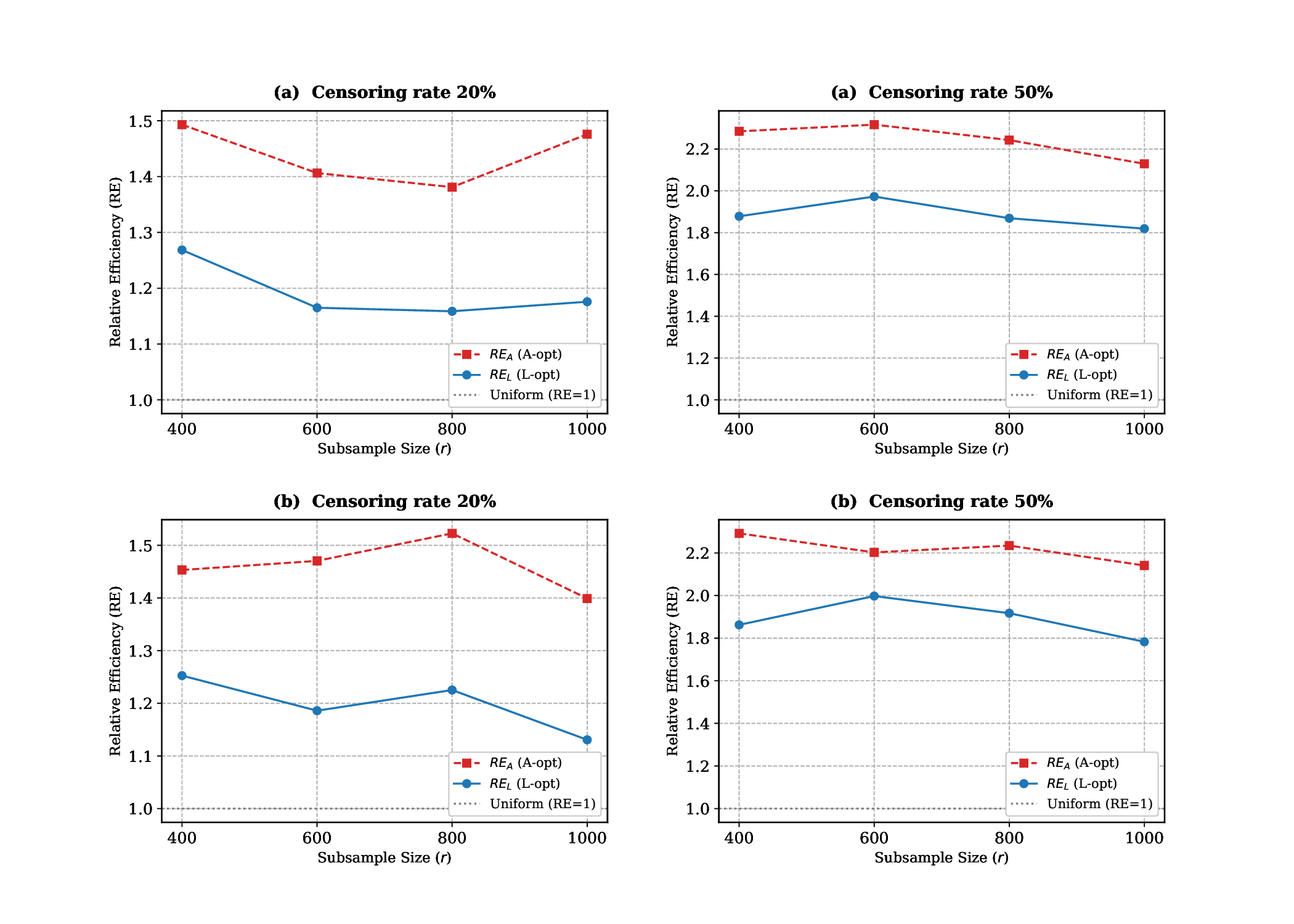}
	\caption{Relative efficiency across different subsample sizes.\label{fig:2}}
\end{figure}
Figure \ref{fig:3} shows that, under the various sampling schemes, the running time decreases substantially as the censoring rate increases. This is due to the reduction in the effective risk set, which lowers the computational complexity associated with risk set updating and weight calculation in the full-sample estimation, thereby significantly reducing the overall runtime. Furthermore, the running times under the two covariate configurations (scenarios (a) and (b)) are nearly identical, indicating that the computational cost is primarily driven by the sample size rather than the specific data distribution.
When comparing the three subsampling methods, uniform sampling incurs the lowest running time, as it does not require constructing complex sampling probabilities. Both the L-opt and A-opt methods involve the calculation of approximately optimal probabilities, which results in a slightly higher computational overhead. Among these, the A-opt method is the most time-consuming, as it entails additional operations involving the matrix $\Psi$; however, the difference between the two optimal methods is negligible. Notably, relative to the substantial computational cost of full-data estimation, the additional cost for probability construction in the L-opt and A-opt strategies is almost negligible. These results indicate that optimal subsampling strategies sacrifice only a modest amount of computational efficiency while achieving substantially improved estimation accuracy compared to uniform sampling, thereby demonstrating strong computational benefits and scalability in large-scale survival data analysis.
\begin{figure}[!h]
	\centering
	\includegraphics[width=0.9\textwidth]{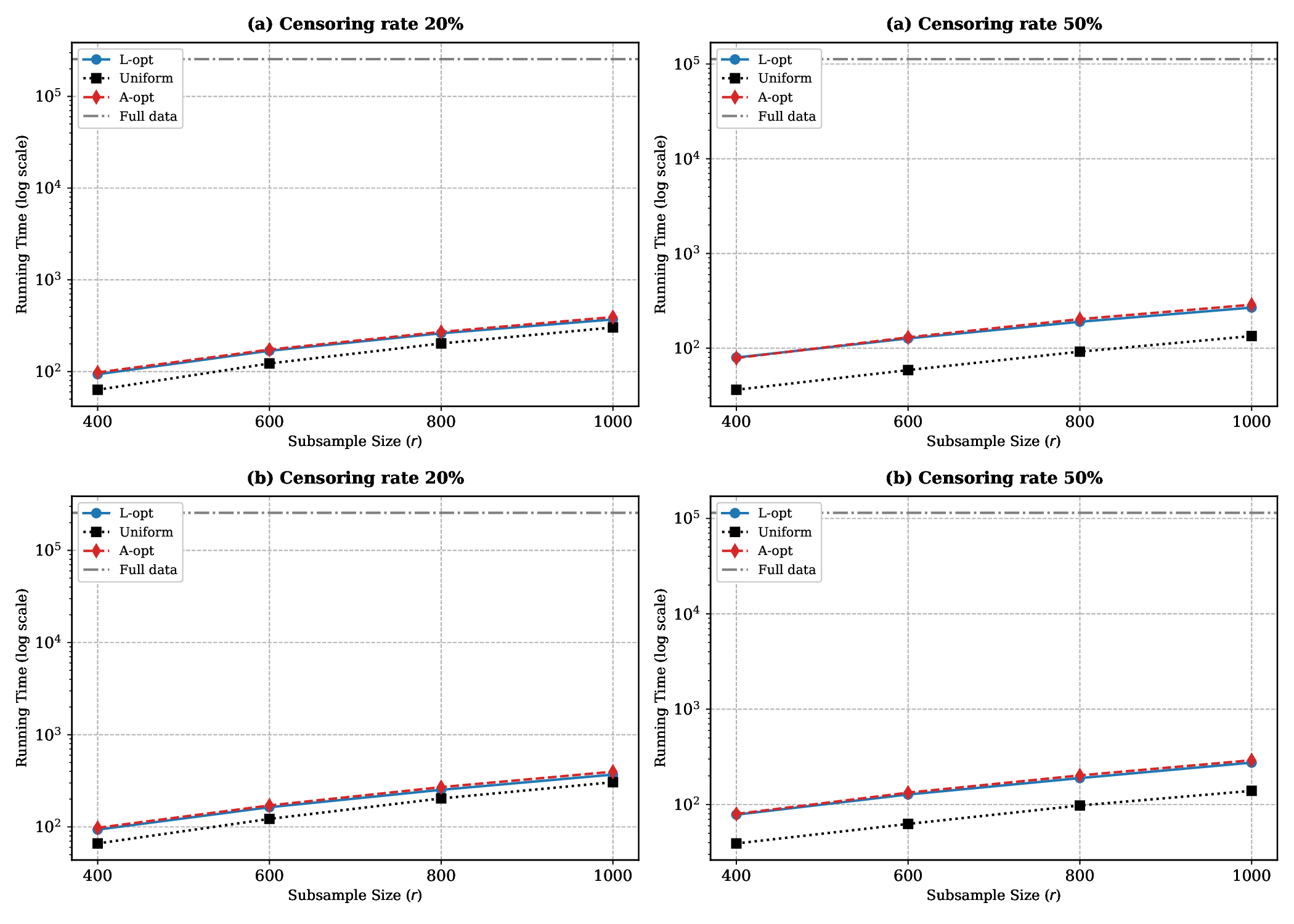}
	\caption{Running time comparison for different subsampling methods.\label{fig:3}}
\end{figure}

Tables \ref{tab:caseI 20} to \ref{tab:caseII 50} present the finite-sample performance of the estimators under various settings, including bias, estimated standard error (ESE), sample standard error (SSE), and 95\% confidence interval coverage probability (CP). The results demonstrate close agreement between the ESE and SSE, indicating that the variance estimator accurately captures the true variability of the estimator. As the subsample size increases, both ESE and SSE exhibit a decreasing trend. Meanwhile, the coverage probabilities remain close to the nominal 95\% level across all settings, further confirming the stability of the proposed method in finite-sample scenarios.
\begin{sidewaystable}[h]
\centering
\caption{(a) censoring rate 20\%}
\label{tab:caseI 20}
\begin{tabular}{c c *{12}{c}}
\toprule
\multirow{2}{*}{$r$} & \multirow{2}{*}{$\boldsymbol{\beta}$} & \multicolumn{4}{c}{Lopt} & \multicolumn{4}{c}{UNIF} & \multicolumn{4}{c}{Aopt} \\
\cmidrule(lr){3-6} \cmidrule(lr){7-10} \cmidrule(lr){11-14}
    & & Bias & ESE & SSE & CP
    & Bias & ESE & SSE & CP
    & Bias & ESE & SSE & CP \\
\midrule
\multirow{5}{*}{400}
 & $\boldsymbol{\beta}_1$ & -0.0214 & 0.1377 & 0.1426 & 0.949 & -0.0242 & 0.1484 & 0.1604 & 0.933 & -0.0078 & 0.1238 & 0.1263 & 0.943 \\
 & $\boldsymbol{\beta}_2$ & -0.0155 & 0.1676 & 0.1729 & 0.942 & -0.0252 & 0.1811 & 0.1843 & 0.943 & -0.0013 & 0.1529 & 0.1569 & 0.937 \\
 & $\boldsymbol{\beta}_3$ & -0.0029 & 0.1593 & 0.1611 & 0.949 & -0.0070 & 0.1728 & 0.1850 & 0.931 & -0.0114 & 0.1472 & 0.1534 & 0.947 \\
 & $\boldsymbol{\beta}_4$ & 0.0098 & 0.0626 & 0.0638 & 0.940 & 0.0134 & 0.0728 & 0.0765 & 0.934 & 0.0091 & 0.0627 & 0.0648 & 0.941 \\
 & $\boldsymbol{\beta}_5$ & 0.0182 & 0.1054 & 0.1083 & 0.944 & 0.0192 & 0.1175 & 0.1288 & 0.923 & 0.0141 & 0.0994 & 0.1018 & 0.945 \\
\midrule
\multirow{5}{*}{600}
 & $\boldsymbol{\beta}_1$ & -0.0091 & 0.1117 & 0.1134 & 0.950 & -0.0049 & 0.1202 & 0.1262 & 0.932 & -0.0017 & 0.1004 & 0.1021 & 0.945 \\
 & $\boldsymbol{\beta}_2$ & -0.0166 & 0.1362 & 0.1400 & 0.942 & -0.0077 & 0.1475 & 0.1505 & 0.948 & -0.0029 & 0.1242 & 0.1258 & 0.953 \\
 & $\boldsymbol{\beta}_3$ & 0.0023 & 0.1296 & 0.1341 & 0.940 & -0.0035 & 0.1405 & 0.1408 & 0.957 & -0.0056 & 0.1193 & 0.1240 & 0.934 \\
 & $\boldsymbol{\beta}_4$ & 0.0040 & 0.0508 & 0.0512 & 0.953 & 0.0013 & 0.0589 & 0.0602 & 0.935 & 0.0059 & 0.0510 & 0.0507 & 0.955 \\
 & $\boldsymbol{\beta}_5$ & 0.0089 & 0.0854 & 0.0879 & 0.942 & 0.0058 & 0.0952 & 0.0962 & 0.959 & 0.0080 & 0.0806 & 0.0800 & 0.951 \\
\midrule
\multirow{5}{*}{800}
 & $\boldsymbol{\beta}_1$ & -0.0061 & 0.0970 & 0.0980 & 0.957 & -0.0125 & 0.1044 & 0.1062 & 0.937 & -0.0045 & 0.0869 & 0.0865 & 0.945 \\
 & $\boldsymbol{\beta}_2$ & -0.0075 & 0.1177 & 0.1180 & 0.947 & -0.0017 & 0.1275 & 0.1247 & 0.956 & -0.0034 & 0.1072 & 0.1075 & 0.951 \\
 & $\boldsymbol{\beta}_3$ & -0.0028 & 0.1118 & 0.1122 & 0.948 & -0.0007 & 0.1219 & 0.1214 & 0.948 & -0.0045 & 0.1030 & 0.1031 & 0.954 \\
 & $\boldsymbol{\beta}_4$ & 0.0031 & 0.0438 & 0.0454 & 0.943 & 0.0060 & 0.0512 & 0.0518 & 0.946 & 0.0017 & 0.0440 & 0.0454 & 0.941 \\
 & $\boldsymbol{\beta}_5$ & 0.0058 & 0.0739 & 0.0763 & 0.944 & 0.0073 & 0.0823 & 0.0811 & 0.945 & 0.0038 & 0.0695 & 0.0727 & 0.942 \\
\midrule
\multirow{5}{*}{1000}
 & $\boldsymbol{\beta}_1$ & -0.0078 & 0.0866 & 0.0892 & 0.940 & -0.0063 & 0.0932 & 0.0934 & 0.943 & -0.0033 & 0.0775 & 0.0772 & 0.952 \\
 & $\boldsymbol{\beta}_2$ & -0.0057 & 0.1054 & 0.1064 & 0.948 & -0.0070 & 0.1143 & 0.1159 & 0.945 & -0.0012 & 0.0957 & 0.0930 & 0.960 \\
 & $\boldsymbol{\beta}_3$ & 0.0020 & 0.1000 & 0.1043 & 0.946 & 0.0043 & 0.1091 & 0.1129 & 0.942 & -0.0044 & 0.0919 & 0.0954 & 0.940 \\
 & $\boldsymbol{\beta}_4$ & 0.0011 & 0.0392 & 0.0395 & 0.946 & 0.0023 & 0.0456 & 0.0470 & 0.942 & 0.0010 & 0.0393 & 0.0392 & 0.952 \\
 & $\boldsymbol{\beta}_5$ & 0.0035 & 0.0661 & 0.0677 & 0.943 & 0.0043 & 0.0734 & 0.0745 & 0.946 & 0.0026 & 0.0622 & 0.0615 & 0.959 \\
\bottomrule
\end{tabular}

\end{sidewaystable}
\begin{sidewaystable}[h]
\centering
\caption{(a) censoring rate 50\%}
\label{tab:caseI 50}
\begin{tabular}{c c *{12}{c}}
\toprule
\multirow{2}{*}{$r$} & \multirow{2}{*}{$\boldsymbol{\beta}$} & \multicolumn{4}{c}{Lopt} & \multicolumn{4}{c}{UNIF} & \multicolumn{4}{c}{Aopt} \\
\cmidrule(lr){3-6} \cmidrule(lr){7-10} \cmidrule(lr){11-14}
    & & Bias & ESE & SSE & CP
    & Bias & ESE & SSE & CP
    & Bias & ESE & SSE & CP \\
\midrule
\multirow{5}{*}{400}
 & $\boldsymbol{\beta}_1$ & -0.0174 & 0.1396 & 0.1435 & 0.949 & -0.0278 & 0.1866 & 0.1912 & 0.939 & -0.0009 & 0.1258 & 0.1271 & 0.942 \\
 & $\boldsymbol{\beta}_2$ & -0.0078 & 0.1788 & 0.1934 & 0.930 & -0.0244 & 0.2429 & 0.2643 & 0.937 & -0.0025 & 0.1657 & 0.1727 & 0.937 \\
 & $\boldsymbol{\beta}_3$ & -0.0007 & 0.1721 & 0.1764 & 0.940 & -0.0103 & 0.2340 & 0.2401 & 0.943 & 0.0003 & 0.1595 & 0.1650 & 0.943 \\
 & $\boldsymbol{\beta}_4$ & 0.0033 & 0.0638 & 0.0665 & 0.942 & 0.0121 & 0.0935 & 0.0990 & 0.934 & -0.0003 & 0.0640 & 0.0651 & 0.940 \\
 & $\boldsymbol{\beta}_5$ & 0.0145 & 0.1053 & 0.1109 & 0.932 & 0.0348 & 0.1479 & 0.1530 & 0.939 & 0.0069 & 0.0989 & 0.1011 & 0.937 \\
\midrule
\multirow{5}{*}{600}
 & $\boldsymbol{\beta}_1$ & -0.0068 & 0.1137 & 0.1135 & 0.949 & -0.0183 & 0.1512 & 0.1576 & 0.937 & 0.0008 & 0.1025 & 0.1015 & 0.955 \\
 & $\boldsymbol{\beta}_2$ & -0.0101 & 0.1454 & 0.1493 & 0.945 & -0.0082 & 0.1967 & 0.2080 & 0.939 & -0.0018 & 0.1346 & 0.1336 & 0.954 \\
 & $\boldsymbol{\beta}_3$ & 0.0001 & 0.1399 & 0.1429 & 0.945 & 0.0046 & 0.1896 & 0.1962 & 0.951 & -0.0046 & 0.1289 & 0.1337 & 0.933 \\
 & $\boldsymbol{\beta}_4$ & 0.0011 & 0.0519 & 0.0527 & 0.944 & 0.0074 & 0.0758 & 0.0800 & 0.940 & -0.0004 & 0.0522 & 0.0533 & 0.938 \\
 & $\boldsymbol{\beta}_5$ & 0.0061 & 0.0857 & 0.0851 & 0.963 & 0.0178 & 0.1193 & 0.1260 & 0.935 & 0.0042 & 0.0805 & 0.0850 & 0.945 \\
\midrule
\multirow{5}{*}{800}
 & $\boldsymbol{\beta}_1$ & -0.0045 & 0.0983 & 0.1019 & 0.953 & -0.0171 & 0.1309 & 0.1361 & 0.937 & 0.0012 & 0.0886 & 0.0883 & 0.941 \\
 & $\boldsymbol{\beta}_2$ & -0.0092 & 0.1261 & 0.1294 & 0.935 & -0.0114 & 0.1701 & 0.1768 & 0.933 & -0.0054 & 0.1163 & 0.1183 & 0.941 \\
 & $\boldsymbol{\beta}_3$ & -0.0019 & 0.1210 & 0.1228 & 0.946 & -0.0048 & 0.1637 & 0.1668 & 0.949 & -0.0039 & 0.1113 & 0.1151 & 0.939 \\
 & $\boldsymbol{\beta}_4$ & -0.0018 & 0.0449 & 0.0459 & 0.938 & 0.0059 & 0.0655 & 0.0667 & 0.947 & -0.0029 & 0.0452 & 0.0446 & 0.945 \\
 & $\boldsymbol{\beta}_5$ & 0.0016 & 0.0740 & 0.0769 & 0.948 & 0.0175 & 0.1031 & 0.1049 & 0.939 & 0.0016 & 0.0697 & 0.0700 & 0.944 \\
\midrule
\multirow{5}{*}{1000}
 & $\boldsymbol{\beta}_1$ & -0.0036 & 0.0882 & 0.0891 & 0.956 & -0.0081 & 0.1166 & 0.1166 & 0.944 & 0.0009 & 0.0792 & 0.0811 & 0.942 \\
 & $\boldsymbol{\beta}_2$ & -0.0092 & 0.1129 & 0.1176 & 0.942 & -0.0029 & 0.1516 & 0.1557 & 0.938 & -0.0051 & 0.1040 & 0.1046 & 0.952 \\
 & $\boldsymbol{\beta}_3$ & -0.0025 & 0.1081 & 0.1080 & 0.956 & -0.0038 & 0.1460 & 0.1476 & 0.947 & -0.0067 & 0.0995 & 0.1025 & 0.940 \\
 & $\boldsymbol{\beta}_4$ & -0.0029 & 0.0402 & 0.0402 & 0.949 & 0.0039 & 0.0586 & 0.0593 & 0.954 & -0.0011 & 0.0403 & 0.0407 & 0.944 \\
 & $\boldsymbol{\beta}_5$ & 0.0028 & 0.0664 & 0.0672 & 0.944 & 0.0103 & 0.0922 & 0.0939 & 0.950 & 0.0025 & 0.0625 & 0.0652 & 0.950 \\
\bottomrule
\end{tabular}
\end{sidewaystable}
\begin{sidewaystable}[h]
\centering
\caption{(b) censoring rate 20\%}
\label{tab:caseII 20}
\begin{tabular}{c c *{12}{c}}
\toprule
\multirow{2}{*}{$r$} & \multirow{2}{*}{$\boldsymbol{\beta}$} & \multicolumn{4}{c}{Lopt} & \multicolumn{4}{c}{UNIF} & \multicolumn{4}{c}{Aopt} \\
\cmidrule(lr){3-6} \cmidrule(lr){7-10} \cmidrule(lr){11-14}
    & & Bias & ESE & SSE & CP
    & Bias & ESE & SSE & CP
    & Bias & ESE & SSE & CP \\
\midrule
\multirow{5}{*}{400}
 & $\boldsymbol{\beta}_1$ & -0.0154 & 0.1360 & 0.1402 & 0.943 & -0.0185 & 0.1474 & 0.1522 & 0.951 & 0.0021 & 0.1208 & 0.1213 & 0.954 \\
 & $\boldsymbol{\beta}_2$ & -0.0024 & 0.0862 & 0.0836 & 0.955 & -0.0042 & 0.0978 & 0.1000 & 0.945 & -0.0004 & 0.0819 & 0.0857 & 0.937 \\
 & $\boldsymbol{\beta}_3$ & -0.0023 & 0.0804 & 0.0832 & 0.948 & -0.0017 & 0.0919 & 0.0968 & 0.942 & 0.0000 & 0.0782 & 0.0825 & 0.932 \\
 & $\boldsymbol{\beta}_4$ & 0.0033 & 0.0939 & 0.0982 & 0.945 & 0.0108 & 0.1059 & 0.1094 & 0.945 & 0.0007 & 0.0886 & 0.0931 & 0.934 \\
 & $\boldsymbol{\beta}_5$ & 0.0083 & 0.1210 & 0.1254 & 0.941 & 0.0125 & 0.1337 & 0.1377 & 0.945 & 0.0002 & 0.1104 & 0.1171 & 0.937 \\
\midrule
\multirow{5}{*}{600}
 & $\boldsymbol{\beta}_1$ & -0.0090 & 0.1111 & 0.1149 & 0.939 & -0.0060 & 0.1192 & 0.1207 & 0.952 & 0.0040 & 0.0985 & 0.0981 & 0.950 \\
 & $\boldsymbol{\beta}_2$ & 0.0011 & 0.0703 & 0.0745 & 0.941 & -0.0070 & 0.0796 & 0.0796 & 0.951 & 0.0037 & 0.0664 & 0.0670 & 0.953 \\
 & $\boldsymbol{\beta}_3$ & -0.0033 & 0.0653 & 0.0637 & 0.956 & 0.0019 & 0.0749 & 0.0807 & 0.924 & -0.0059 & 0.0634 & 0.0654 & 0.944 \\
 & $\boldsymbol{\beta}_4$ & 0.0019 & 0.0766 & 0.0769 & 0.943 & 0.0037 & 0.0867 & 0.0885 & 0.939 & 0.0011 & 0.0724 & 0.0729 & 0.953 \\
 & $\boldsymbol{\beta}_5$ & 0.0053 & 0.0987 & 0.1040 & 0.943 & 0.0094 & 0.1087 & 0.1075 & 0.951 & -0.0014 & 0.0901 & 0.0916 & 0.950 \\
\midrule
\multirow{5}{*}{800}
 & $\boldsymbol{\beta}_1$ & 0.0002 & 0.0962 & 0.1005 & 0.936 & -0.0085 & 0.1041 & 0.1105 & 0.932 & 0.0066 & 0.0852 & 0.0834 & 0.950 \\
 & $\boldsymbol{\beta}_2$ & -0.0015 & 0.0610 & 0.0621 & 0.949 & -0.0037 & 0.0689 & 0.0724 & 0.938 & 0.0015 & 0.0574 & 0.0589 & 0.943 \\
 & $\boldsymbol{\beta}_3$ & 0.0005 & 0.0566 & 0.0596 & 0.938 & -0.0004 & 0.0647 & 0.0660 & 0.954 & -0.0016 & 0.0549 & 0.0583 & 0.933 \\
 & $\boldsymbol{\beta}_4$ & 0.0006 & 0.0665 & 0.0687 & 0.943 & 0.0016 & 0.0748 & 0.0751 & 0.945 & -0.0008 & 0.0627 & 0.0642 & 0.952 \\
 & $\boldsymbol{\beta}_5$ & 0.0009 & 0.0854 & 0.0891 & 0.947 & 0.0116 & 0.0945 & 0.0960 & 0.949 & -0.0031 & 0.0781 & 0.0796 & 0.942 \\
\midrule
\multirow{5}{*}{1000}
 & $\boldsymbol{\beta}_1$ & -0.0013 & 0.0860 & 0.0924 & 0.925 & -0.0066 & 0.0925 & 0.0908 & 0.961 & 0.0039 & 0.0763 & 0.0786 & 0.945 \\
 & $\boldsymbol{\beta}_2$ & -0.0017 & 0.0547 & 0.0565 & 0.943 & 0.0012 & 0.0614 & 0.0613 & 0.955 & 0.0034 & 0.0516 & 0.0530 & 0.941 \\
 & $\boldsymbol{\beta}_3$ & 0.0001 & 0.0506 & 0.0505 & 0.950 & -0.0033 & 0.0579 & 0.0569 & 0.953 & -0.0017 & 0.0491 & 0.0523 & 0.930 \\
 & $\boldsymbol{\beta}_4$ & -0.0003 & 0.0595 & 0.0631 & 0.939 & 0.0044 & 0.0668 & 0.0697 & 0.938 & 0.0002 & 0.0562 & 0.0579 & 0.938 \\
 & $\boldsymbol{\beta}_5$ & -0.0008 & 0.0765 & 0.0789 & 0.948 & 0.0044 & 0.0839 & 0.0864 & 0.944 & -0.0047 & 0.0699 & 0.0688 & 0.957 \\
\bottomrule
\end{tabular}
\end{sidewaystable}
\begin{sidewaystable}[h]
\centering
\caption{(b) censoring rate 50\%}
\label{tab:caseII 50}
\begin{tabular}{c c *{12}{c}}
\toprule
\multirow{2}{*}{$r$} & \multirow{2}{*}{$\boldsymbol{\beta}$} & \multicolumn{4}{c}{Lopt} & \multicolumn{4}{c}{UNIF} & \multicolumn{4}{c}{Aopt} \\
\cmidrule(lr){3-6} \cmidrule(lr){7-10} \cmidrule(lr){11-14}
    & & Bias & ESE & SSE & CP
    & Bias & ESE & SSE & CP
    & Bias & ESE & SSE & CP \\
\midrule
\multirow{5}{*}{400}
 & $\boldsymbol{\beta}_1$ & -0.0066 & 0.1398 & 0.1473 & 0.943 & -0.0271 & 0.1869 & 0.1960 & 0.935 & -0.0027 & 0.1259 & 0.1285 & 0.944 \\
 & $\boldsymbol{\beta}_2$ & 0.0029 & 0.0915 & 0.0970 & 0.928 & -0.0131 & 0.1293 & 0.1373 & 0.933 & -0.0002 & 0.0869 & 0.0872 & 0.949 \\
 & $\boldsymbol{\beta}_3$ & -0.0006 & 0.0871 & 0.0926 & 0.931 & 0.0012 & 0.1241 & 0.1310 & 0.939 & 0.0078 & 0.0832 & 0.0887 & 0.926 \\
 & $\boldsymbol{\beta}_4$ & -0.0044 & 0.0999 & 0.1071 & 0.932 & 0.0133 & 0.1404 & 0.1425 & 0.946 & -0.0102 & 0.0942 & 0.0949 & 0.946 \\
 & $\boldsymbol{\beta}_5$ & 0.0104 & 0.1241 & 0.1306 & 0.937 & 0.0258 & 0.1710 & 0.1755 & 0.942 & 0.0024 & 0.1146 & 0.1201 & 0.944 \\
\midrule
\multirow{5}{*}{600}
 & $\boldsymbol{\beta}_1$ & 0.0015 & 0.1141 & 0.1154 & 0.949 & -0.0228 & 0.1511 & 0.1608 & 0.925 & 0.0034 & 0.1027 & 0.1120 & 0.932 \\
 & $\boldsymbol{\beta}_2$ & -0.0004 & 0.0747 & 0.0790 & 0.936 & -0.0134 & 0.1041 & 0.1089 & 0.925 & 0.0016 & 0.0708 & 0.0716 & 0.951 \\
 & $\boldsymbol{\beta}_3$ & 0.0036 & 0.0711 & 0.0715 & 0.952 & 0.0017 & 0.1004 & 0.1090 & 0.926 & 0.0024 & 0.0675 & 0.0680 & 0.950 \\
 & $\boldsymbol{\beta}_4$ & -0.0027 & 0.0818 & 0.0814 & 0.952 & 0.0139 & 0.1138 & 0.1215 & 0.940 & -0.0044 & 0.0770 & 0.0801 & 0.947 \\
 & $\boldsymbol{\beta}_5$ & 0.0063 & 0.1017 & 0.1082 & 0.936 & 0.0245 & 0.1384 & 0.1415 & 0.941 & 0.0032 & 0.0940 & 0.1017 & 0.942 \\
\midrule
\multirow{5}{*}{800}
 & $\boldsymbol{\beta}_1$ & -0.0022 & 0.0995 & 0.1030 & 0.943 & -0.0194 & 0.1309 & 0.1372 & 0.944 & 0.0038 & 0.0895 & 0.0926 & 0.940 \\
 & $\boldsymbol{\beta}_2$ & -0.0007 & 0.0651 & 0.0695 & 0.931 & -0.0089 & 0.0905 & 0.0944 & 0.928 & -0.0023 & 0.0615 & 0.0634 & 0.947 \\
 & $\boldsymbol{\beta}_3$ & 0.0051 & 0.0620 & 0.0639 & 0.942 & 0.0003 & 0.0871 & 0.0941 & 0.925 & 0.0048 & 0.0586 & 0.0612 & 0.944 \\
 & $\boldsymbol{\beta}_4$ & -0.0051 & 0.0711 & 0.0699 & 0.942 & 0.0052 & 0.0987 & 0.1058 & 0.934 & -0.0061 & 0.0668 & 0.0715 & 0.940 \\
 & $\boldsymbol{\beta}_5$ & 0.0027 & 0.0886 & 0.0969 & 0.927 & 0.0198 & 0.1201 & 0.1270 & 0.937 & 0.0032 & 0.0820 & 0.0869 & 0.937 \\
\midrule
\multirow{5}{*}{1000}
 & $\boldsymbol{\beta}_1$ & -0.0033 & 0.0891 & 0.0948 & 0.938 & -0.0116 & 0.1168 & 0.1227 & 0.941 & 0.0032 & 0.0802 & 0.0824 & 0.945 \\
 & $\boldsymbol{\beta}_2$ & 0.0015 & 0.0582 & 0.0611 & 0.934 & -0.0061 & 0.0805 & 0.0858 & 0.939 & 0.0014 & 0.0552 & 0.0573 & 0.947 \\
 & $\boldsymbol{\beta}_3$ & 0.0014 & 0.0555 & 0.0583 & 0.934 & -0.0017 & 0.0778 & 0.0793 & 0.946 & 0.0042 & 0.0526 & 0.0560 & 0.938 \\
 & $\boldsymbol{\beta}_4$ & -0.0039 & 0.0637 & 0.0663 & 0.945 & 0.0092 & 0.0882 & 0.0925 & 0.940 & -0.0053 & 0.0598 & 0.0634 & 0.935 \\
 & $\boldsymbol{\beta}_5$ & 0.0016 & 0.0795 & 0.0873 & 0.935 & 0.0149 & 0.1067 & 0.1110 & 0.939 & 0.0017 & 0.0737 & 0.0784 & 0.939 \\
\bottomrule
\end{tabular}
\end{sidewaystable}
\section{Airline flight data analysis}
In this section, we empirically evaluate the proposed subsampling methods for Cox regression with length-biased right-censored data using the 2007 records from the Data Expo 2009: Airline on time data \citep{DVN/HG7NV7_2008} archive, which contains flight operation data from U.S. airlines. The original dataset contains 7,453,215 flight records. After excluding canceled flights, diverted flights, observations with missing values for key variables, and records with abnormal operating times, 7,202,596 valid flight records remain. To mimic a length-biased sampling mechanism, we draw 500,000 records with replacement using probabilities proportional to the actual flight operating time $T$, that is, $p_i = T_i / \sum_{j=1}^{N} T_j$. Under this sampling design, flights with longer operating times have higher inclusion probabilities, and the resulting sample is used as the full dataset for the empirical analysis.

Following this selection procedure, for each selected record, we generate $U\sim \mathrm{Uniform}(0,1)$ and construct the entry time $A=UT$ and the residual lifetime $V=T-A$. The censoring time is then artificially set to $C=90$ minutes. Accordingly, we define the observed residual lifetime $V_{\mathrm{obs}}=\min(V,C)$, the observed time $Y=A+V_{\mathrm{obs}}$, and the censoring indicator $\delta=I(V\le C)$. This procedure yields a censoring rate of 33.8\% and produces the length-biased right-censored dataset used for the empirical analysis.

In the Cox model, we include flight type, travel date, season, and scheduled departure period as covariates. Specifically, $x_1$ = 1 and $x_2$ = 1 indicate long-haul flights and weekend flights, respectively. Season is encoded using three dummy variables, with spring as the reference group:  $x_3$ = 1, $x_4$ = 1 and  $x_5$ = 1 indicate that the flight occurs in summer, autumn, and winter, respectively. Similarly, the scheduled departure period is encoded using two dummy variables, with morning as the reference group: $x_6$ = 1 and  $x_7$ = 1 denote afternoon and evening departures, respectively. To evaluate the estimation performance of the subsampling methods, we first fit the Cox model to the full dataset described above. The full-data estimate of the regression coefficient vector is  $\hat{\boldsymbol{\beta}}_{full}=(-2.4313,0.0472,-0.0462,0.0012,-0.0791,-0.0463,-0.0048)^\top.$  The full-sample estimation results indicate that differences in flight operating time are primarily captured by the effect of flight type: long-haul flights generally have longer actual operating times than short-haul flights. In contrast, the effects of variables such as weekend travel, season, and scheduled departure period are relatively limited. Specifically, weekend flights are slightly shorter than weekday flights, winter flights are slightly longer than spring flights, and the departure-period variables exhibit relatively weak statistical effects.

Based on these results, we further compare the performance of different subsampling methods for estimation of the Cox model. Table \ref{tab9} reports the bias, average standard error (SE), and empirical standard error (ESE) for each subsampling method based on 1,000 repeated subsamples. The biases of all methods are close to zero, indicating that the subsampling estimators are generally consistent with the full-data estimates and exhibit no systematic bias. Furthermore, the SE and ESE values are broadly comparable, suggesting that the estimated variances adequately capture the sampling variability of the estimators. A closer comparison reveals that the A-optimal and L-optimal methods yield smaller standard errors than uniform subsampling for most parameters, implying that the optimal subsampling approaches achieve higher estimation efficiency. These findings are broadly consistent with the conclusions drawn from the preceding simulation studies and provide further evidence that A-optimal and L-optimal subsampling outperform uniform subsampling in reducing estimation variability. In terms of computational cost, fitting the Cox model to the full dataset requires 51,230.51 seconds, whereas even the most computationally intensive subsampling method considered, A-optimal, takes only 365.96 seconds with a subsample size of r=1000. This indicates that subsampling strategies can substantially reduce the computational burden associated with full-data Cox model estimation.
\begin{table}[htbp]
\centering
\caption{Estimation results for U.S. airline flight operation data}
\label{tab9}
\begin{tabular}{c c *{9}{c}}
\toprule
\multirow{2}{*}{$r$} & \multirow{2}{*}{$\boldsymbol{\beta}$}
& \multicolumn{3}{c}{L-opt} & \multicolumn{3}{c}{A-opt} & \multicolumn{3}{c}{Uniform} \\
\cmidrule(lr){3-5} \cmidrule(lr){6-8} \cmidrule(lr){9-11}
 & & Bias & ESE & SE & Bias & ESE & SE & Bias & ESE & SE \\
\midrule
\multirow{7}{*}{400}
 & $\boldsymbol{\beta}_1$ & -0.1097 & 0.1798 & 0.2261 & -0.0906 & 0.1707 & 0.2050 & -0.1892 & 0.2166 & 0.3102 \\
 & $\boldsymbol{\beta}_2$ & -0.0016 & 0.1160 & 0.1242 & -0.0021 & 0.1169 & 0.1210 & 0.0038 & 0.1463 & 0.1617 \\
 & $\boldsymbol{\beta}_3$ & 0.0015 & 0.1554 & 0.1656 & 0.0115 & 0.1440 & 0.1580 & 0.0068 & 0.1867 & 0.2101 \\
 & $\boldsymbol{\beta}_4$ & -0.0067 & 0.1548 & 0.1629 & 0.0104 & 0.1487 & 0.1646 & 0.0038 & 0.1845 & 0.1994 \\
 & $\boldsymbol{\beta}_5$ & -0.0081 & 0.1593 & 0.1694 & 0.0102 & 0.1497 & 0.1605 & -0.0072 & 0.1905 & 0.2179 \\
 & $\boldsymbol{\beta}_6$ & -0.0017 & 0.1235 & 0.1378 & -0.0002 & 0.1224 & 0.1335 & 0.0066 & 0.1482 & 0.1720 \\
 & $\boldsymbol{\beta}_7$ & 0.0030 & 0.1422 & 0.1568 & 0.0030 & 0.1400 & 0.1503 & 0.0064 & 0.1787 & 0.2091 \\
\midrule
\multirow{7}{*}{600}
 & $\boldsymbol{\beta}_1$ & -0.0750 & 0.1516 & 0.1829 & -0.0681 & 0.1437 & 0.1623 & -0.1219 & 0.1843 & 0.2427 \\
 & $\boldsymbol{\beta}_2$ & -0.0047 & 0.0955 & 0.1007 & -0.0070 & 0.0963 & 0.1014 & 0.0031 & 0.1214 & 0.1352 \\
 & $\boldsymbol{\beta}_3$ & -0.0030 & 0.1286 & 0.1373 & 0.0054 & 0.1183 & 0.1236 & 0.0004 & 0.1558 & 0.1725 \\
 & $\boldsymbol{\beta}_4$ & -0.0041 & 0.1278 & 0.1369 & 0.0050 & 0.1222 & 0.1271 & -0.0031 & 0.1519 & 0.1652 \\
 & $\boldsymbol{\beta}_5$ & -0.0014 & 0.1313 & 0.1425 & 0.0091 & 0.1225 & 0.1266 & -0.0113 & 0.1574 & 0.1766 \\
 & $\boldsymbol{\beta}_6$ & 0.0036 & 0.1022 & 0.1085 & 0.0087 & 0.1006 & 0.1073 & 0.0048 & 0.1236 & 0.1360 \\
 & $\boldsymbol{\beta}_7$ & 0.0027 & 0.1174 & 0.1262 & 0.0055 & 0.1153 & 0.1241 & -0.0002 & 0.1477 & 0.1597 \\
\midrule
\multirow{7}{*}{800}
 & $\boldsymbol{\beta}_1$ & -0.0528 & 0.1339 & 0.1539 & -0.0489 & 0.1264 & 0.1413 & -0.1043 & 0.1655 & 0.2000 \\
 & $\boldsymbol{\beta}_2$ & -0.0027 & 0.0832 & 0.0872 & -0.0040 & 0.0835 & 0.0860 & -0.0010 & 0.1067 & 0.1157 \\
 & $\boldsymbol{\beta}_3$ & 0.0026 & 0.1127 & 0.1179 & 0.0016 & 0.1032 & 0.1076 & 0.0021 & 0.1364 & 0.1546 \\
 & $\boldsymbol{\beta}_4$ & -0.0002 & 0.1118 & 0.1181 & 0.0023 & 0.1062 & 0.1123 & -0.0048 & 0.1336 & 0.1451 \\
 & $\boldsymbol{\beta}_5$ & 0.0029 & 0.1152 & 0.1202 & 0.0078 & 0.1065 & 0.1104 & -0.0079 & 0.1388 & 0.1458 \\
 & $\boldsymbol{\beta}_6$ & 0.0054 & 0.0891 & 0.0943 & 0.0056 & 0.0874 & 0.0910 & 0.0008 & 0.1087 & 0.1180 \\
 & $\boldsymbol{\beta}_7$ & 0.0004 & 0.1028 & 0.1094 & -0.0028 & 0.1006 & 0.1061 & 0.0111 & 0.1290 & 0.1460 \\
\midrule
\multirow{7}{*}{1000}
 & $\boldsymbol{\beta}_1$ & -0.0397 & 0.1217 & 0.1340 & -0.0359 & 0.1140 & 0.1311 & -0.0828 & 0.1521 & 0.1888 \\
 & $\boldsymbol{\beta}_2$ & -0.0014 & 0.0749 & 0.0784 & 0.0030 & 0.0748 & 0.0793 & -0.0034 & 0.0961 & 0.1049 \\
 & $\boldsymbol{\beta}_3$ & -0.0034 & 0.1017 & 0.1053 & 0.0024 & 0.0926 & 0.0952 & -0.0030 & 0.1235 & 0.1402 \\
 & $\boldsymbol{\beta}_4$ & -0.0019 & 0.1003 & 0.1060 & 0.0037 & 0.0953 & 0.0994 & -0.0114 & 0.1204 & 0.1283 \\
 & $\boldsymbol{\beta}_5$ & -0.0040 & 0.1034 & 0.1057 & 0.0036 & 0.0958 & 0.1010 & -0.0080 & 0.1247 & 0.1347 \\
 & $\boldsymbol{\beta}_6$ & 0.0062 & 0.0802 & 0.0860 & 0.0039 & 0.0786 & 0.0858 & 0.0044 & 0.0980 & 0.1088 \\
 & $\boldsymbol{\beta}_7$ & 0.0040 & 0.0923 & 0.0943 & 0.0055 & 0.0900 & 0.0932 & -0.0005 & 0.1173 & 0.1252 \\
\bottomrule
\end{tabular}
\end{table}
\section{Conclusion} %
In this paper, we have developed a novel and computationally efficient subsampling framework for the Cox proportional hazards model when applied to large-scale length-biased right-censored survival data. The proposed method addresses the growing computational challenges posed by massive survival datasets under complex sampling mechanisms.

Theoretically, we established the consistency and asymptotic normality of the subsample-based estimator, demonstrating that it provides an efficient approximation to the full-data estimator under mild regularity conditions. Furthermore, we derived explicit optimal subsampling probabilities under the A- and L-optimality criteria, which minimize the asymptotic variance of the resulting estimator. These optimal designs enable the subsample estimator to approximate the full-sample estimator with high fidelity while substantially reducing computational overhead.

Numerical simulations were conducted to evaluate the finite-sample performance of the proposed method under various configurations, including different censoring rates, subsample sizes, and covariate settings. The simulation results corroborate our theoretical findings: the optimal subsampling methods, particularly A-optimal and L-optimal, yield nearly unbiased estimates with smaller standard errors than uniform subsampling and estimates close to those obtained from the full-data analysis, while significantly outperforming uniform subsampling in terms of estimation efficiency. An empirical analysis using the large-scale U.S. airline on-time performance dataset further validated the practical utility of the proposed approach. Specifically, subsampling reduced the computational time from over 51,000 seconds for the full-data analysis to less than 370 seconds for the A-optimal method with a subsample size of r=1000, achieving a computational speedup of two orders of magnitude with negligible loss in estimation accuracy.

In summary, the proposed subsampling framework provides a theoretically justified and practically feasible strategy for efficient inference in very large-scale biased survival data. It strikes a favorable balance between statistical accuracy and computational cost, offering a scalable solution for modern survival analysis.

Future work may extend this approach in several directions. First, it can be adapted to settings with missing covariates by integrating subsampling with missing data handling techniques, such as inverse probability weighting or multiple imputation. Second, the proposed methodology can be generalized to other survival models to enhance robustness and broaden applicability. Promising candidates include the accelerated failure time model, models with time-dependent covariates, and frailty models for clustered or heterogeneous survival data. Third, further investigation into adaptive or sequential subsampling strategies that dynamically refine the sampling distribution based on pilot estimates may yield additional gains in efficiency and robustness. These extensions would further enhance the method's applicability and generalizability across a wider range of large-scale survival analysis problems.
\section*{Disclosure statement}
The authors have no relevant financial or non-financial interests to disclose.
\section*{Funding}
Jianhua SHI's research was supported by the Natural Science Foundation of Fujian Province (No. 2025J01969), China; Lin OUYANG's research was supported by the Key Science and Technology Plan Project of Zhangzhou City (No. ZZ20250012), China.
\section*{Data availability statement}
The raw data analyzed in this study are publicly available from the Harvard Dataverse repository at \url{https://doi.org/10.7910/DVN/HG7NV7}.

\bigskip
\appendix
\section{Proofs}
Before proving Theorem \ref{th1}, it is necessary to establish Lemma \ref{lem1}.
\begin{lemma}\label{lem1}\normalsize
    Under Assumptions \ref{ass2} and \ref{ass4}, conditional on the full sample data $\mathcal{D}_n$, as $r \to \infty$, we have
\begin{equation}
    \sup_{\boldsymbol{\beta}} \left\| U^*(\boldsymbol{\beta}) - U(\boldsymbol{\beta}) \right\| \xrightarrow{P} 0,
\end{equation}
where $\xrightarrow{P}$ denotes convergence in probability.
\end{lemma}
\begin{proof}[Proof.]
    \normalsize
    For a subsample of size $r$, we define the subsample empirical measure conditional on the full data $\mathcal{D}_n$ as
\begin{equation}
    \mathbb{P}_{r|\mathcal{D}_n} = \frac{1}{r} \sum_{i=1}^{r} \delta_{Z_i^*}.
\end{equation}
where, \(\delta_{Z_i^*}\) is a measure that assigns mass 1 at the point \(Z_i^*\) and 0 elsewhere. The random variable \(Z_i^*\) is independently and identically distributed, drawn from the full data according to probabilities \(\pi_i\). For a measurable function \(f: \mathcal{D}_n \to \mathbb{R}\), we may denote
\begin{equation}
    \mathbb{P}_{r \mid \mathcal{D}_n} f = \frac{1}{r}\sum_{i=1}^r f(Z_i^*).
\end{equation}
Taking expectation conditional on the full data yields
\begin{equation}
    \mathbf{P}_{\mathcal{D}_n} f(Z^*) = E\big\{ f(Z^*) \mid \mathcal{D}_n \big\} = \sum_{i=1}^n \pi_i f(Z_i).
\end{equation}
Define
\begin{equation}
    f^{(k)}_{t,\boldsymbol{\beta}}(Z^*) = \big(n\pi^*\big)^{-1} w_c^*(t)\, R^*(t)\, W^*\{\mathbf{X}^{*\otimes k}\} \exp(\boldsymbol{\beta}^\top \mathbf{X}^*).
\end{equation}
Then,
\begin{equation}
    S^{*(k)}(t, \boldsymbol{\beta}) = \mathbb{P}_{r \mid \mathcal{D}_n} \, f_{t,\boldsymbol{\beta}}^{(k)}.
\end{equation}
By the definition of conditional expectation,
\begin{equation}
\begin{aligned}
\mathbf{P}_{\mathcal{D}_n} f_{t,\boldsymbol{\beta}}^{(k)}
&= \sum_{i=1}^{n} \pi_i f_{t,\boldsymbol{\beta}}^{(k)} \\
&= n^{-1} \sum_{i=1}^{n} w_c(t)\, R_i(t)\, W_i\, \mathbf{X}_i^{\otimes k} \exp\!\left( \boldsymbol{\beta}^{\top} \mathbf{X}_i \right) \\
&= S^{(k)}(t, \boldsymbol{\beta}).
\end{aligned}
\end{equation}
Therefore,
\begin{equation}
S^{*(k)}(t, \boldsymbol{\beta}) - S^{(k)}(t, \boldsymbol{\beta})
= \left( \mathbb{P}_{r \mid \mathcal{D}_n} - \mathbf{P}_{\mathcal{D}_n} \right) f_{t,\boldsymbol{\beta}}^{(k)}.
\end{equation}
According to \cite{kosorok2008introduction}
 and Assumptions \ref{ass3} and \ref{ass4}, we know that $\{(n\pi)^{-1}N(t): t\in[0,\tau]\}$ and $(n\pi)^{-1}w_c(t)R_i(t)W_i\mathbf{X}_i^{\otimes k}\exp(\boldsymbol{\beta}^{\top}\mathbf{X}_i)$ are Donsker classes. Under the condition $\mathcal{D}_n$, the empirical process converges weakly to a Brownian bridge process
\begin{equation}
\sqrt{r}(\mathbb{P}_{r|\mathcal{D}_n}-\mathbf{P}_{\mathcal{D}_n})f_{t,\boldsymbol{\beta}}^{(k)}\Rightarrow \mathbb{G}_{\mathcal{D}_n}^{(k)},
\end{equation}
in $\ell^\infty([0,\tau]\times\Theta)$. Therefore,
\begin{equation}
\sup_{t\in[0,\tau],\boldsymbol{\beta}\in\Theta}|S^{*(k)}(t,\boldsymbol{\beta})-S^{(k)}(t,\boldsymbol{\beta})|=O_p(r^{-1/2}\mid \mathcal{D}_n).
\end{equation}
In particular,
\begin{equation}
\sup_{t,\boldsymbol{\beta}}|S^{*(k)}(t,\boldsymbol{\beta})-S^{(k)}(t,\boldsymbol{\beta})|\xrightarrow{P|\mathcal{D}_n}0.
\end{equation}
Note that $\bar{X}_w^*(t,\boldsymbol{\beta})=\dfrac{S^{*(1)}(t,\boldsymbol{\beta})}{S^{*(0)}(t,\boldsymbol{\beta})}$, and that $S^{*(1)}(t,\boldsymbol{\beta})$ and $S^{*(0)}(t,\boldsymbol{\beta})$ converge uniformly, while $S^{(0)}(t,\boldsymbol{\beta})$ is bounded away from zero, it follows from the Continuous Mapping Theorem that, conditional on $\mathcal{D}_n$ and as $r\to\infty$, we have
\begin{equation}\label{S1}
    \sup_{t\in[0,\tau],\boldsymbol{\beta}\in\Theta}\mid\overline{\mathbf{X}}_w^*(t,\boldsymbol{\beta})-\overline{\mathbf{X}}_w(t,\boldsymbol{\beta})\mid\xrightarrow{P|\mathcal{D}_n}0.
\end{equation}
The subsample score and the whole sample score are respectively defined as
$$U^{*}(\boldsymbol{\beta})=\mathbb{P}_{r|\mathcal{D}_n} {\psi}_{\boldsymbol{\beta}}^{*},\quad U(\boldsymbol{\beta})=\mathbf{P}_{\mathcal{D}_n} {\psi}_{\boldsymbol{\beta}},$$
where
\begin{equation}
\begin{gathered}
\psi_{\boldsymbol{\beta}}^{*}=\frac{1}{n\pi}\int_{0}^{\tau}\{\mathbf{X}-\bar{\mathbf{X}}_{w}^{*}(t,\boldsymbol{\beta})\}dN(t),\\\psi_{\boldsymbol{\beta}}=\frac{1}{n\pi}\int_{0}^{\tau}\{\mathbf{X}-\bar{\mathbf{X}}_{w}(t,\boldsymbol{\beta})\}dN(t).
\end{gathered}
\end{equation}
We now decompose the difference as follows:
\begin{equation}
    \begin{aligned}
   {{U}^{*}}(\boldsymbol{\beta} )-U(\boldsymbol{\beta} )&={{\mathbb{P}}_{r|{{\mathcal{D}}_{n}}}}\psi _{\boldsymbol{\beta}}^{*}-{{\mathbf{P}}_{{\mathcal{D}_{n}}}}{{\psi }_{\boldsymbol{\beta}}} \\
 & =\underbrace{({{\mathbb{P}}_{r|{{\mathcal{D}}_{n}}}}-{{\mathbf{P}}_{{\mathcal{D}_{n}}}}){{\psi }_{\boldsymbol{\beta}}}}_{I}+\underbrace{({{\mathbb{P}}_{r|{{\mathcal{D}}_{n}}}}-{{\mathbf{P}}_{{\mathcal{D}_{n}}}})(\psi _{\boldsymbol{\beta}}^{*}-{{\psi }_{\boldsymbol{\beta}}})}_{II}+\underbrace{{{\mathbf{P}}_{{\mathcal{D}_{n}}}}(\psi _{\boldsymbol{\beta} }^{*}-{{\psi }_{\boldsymbol{\beta}}})}_{III}.
\end{aligned}
\end{equation}
For the first term I, under the full sample $\mathcal{D}_n$, the function $\overline{\mathbf{X}}_w(t,\boldsymbol{\beta}) = \frac{S^{(1)}(t,\boldsymbol{\beta})}{S^{(0)}(t,\boldsymbol{\beta})}$ is non-random. Since the covariate $\mathbf{X}$ is bounded and the denominator $S^{(0)}(t,\boldsymbol{\beta})$ is bounded away from zero, it follows that $\overline{\mathbf{X}}_w(t,\boldsymbol{\beta})$ is continuously differentiable about $\boldsymbol{\beta}$, with uniformly bounded derivative on the compact set $\Theta \times [0,\tau]$.

Therefore, there exists a constant $L > 0$ such that for any $\boldsymbol{\beta}_1, \boldsymbol{\beta}_2 \in \Theta$ and all $t \in [0,\tau]$,
\begin{equation}
\left\| \overline{\mathbf{X}}_w(t,\boldsymbol{\beta}_1) - \overline{\mathbf{X}}_w(t,\boldsymbol{\beta}_2) \right\|
\leq L \left\| \boldsymbol{\beta}_1 - \boldsymbol{\beta}_2 \right\|.
\end{equation}
Then,
\begin{equation}
\begin{aligned}
\left| \psi_{\boldsymbol{\beta}_1}(Z) - \psi_{\boldsymbol{\beta}_2}(Z) \right|
&= \left| \int_{0}^{\tau} \left\{ \overline{\mathbf{X}}_w(t,\boldsymbol{\beta}_2) - \overline{\mathbf{X}}_w(t,\boldsymbol{\beta}_1) \right\} \frac{1}{n\pi} \, dN(t) \right| \\
&\leq \int_{0}^{\tau} \left\| \overline{\mathbf{X}}_w(t,\boldsymbol{\beta}_2) - \overline{\mathbf{X}}_w(t,\boldsymbol{\beta}_1) \right\| \frac{1}{n\pi} \, dN(t) \\
&\leq \left( L \int_{0}^{\tau} \frac{1}{n\pi} \, dN(t) \right) \left\| \boldsymbol{\beta}_1 - \boldsymbol{\beta}_2 \right\| \\
&= M(Z) \left\| \boldsymbol{\beta}_1 - \boldsymbol{\beta}_2 \right\|.
\end{aligned}
\end{equation}
Here, the random variable $M(Z) = \frac{L}{n\pi} N(\tau)$ serves as a Lipschitz constant for the function class. Given the full sample $\mathcal{D}_n$ and Assumption 3, the process $\{(n\pi)^{-1}N(t) : t \in [0, \tau]\}$ is Donsker and has a bounded second moment under the empirical measure $\mathbf{P}_n$. Hence, the second moment of $M(Z)$ exists and is finite, i.e.,
\begin{equation}
\mathbf{P}_n \left[ M(Z)^2 \right] < \infty.
\end{equation}

By empirical process theory \citep{vaartvanderWeakConvergenceEmpirical1996}
, for a parametric function class indexed by $\boldsymbol{\beta}$ over a compact set $\Theta$, if it satisfies a random Lipschitz condition about $\boldsymbol{\beta}$ and admits a square-integrable envelope, then $\mathcal{F}_n = \left\{ \psi_{\boldsymbol{\beta}} : \boldsymbol{\beta} \in \Theta \right\}$ is a Donsker class.

Therefore,
\begin{equation}
\sup_{\boldsymbol{\beta}} \left\| I \right\|
= O_{P \mid \mathcal{D}_n} \left( r^{-1/2} \right)
= o_{P \mid \mathcal{D}_n} (1).
\end{equation}
For the remaining terms, we have
\begin{equation}
\begin{aligned}
\sup_{\boldsymbol{\beta}} \left\| II \right\|
&\leq \sup_{t,\boldsymbol{\beta}} \left\| \overline{\mathbf{X}}^*(t,\boldsymbol{\beta}) - \overline{\mathbf{X}}(t,\boldsymbol{\beta}) \right\|\cdot (\mathbb{P}_{r \mid \mathcal{D}_n} - \mathbf{P}_{\mathcal{D}_n})
\left[ \int_0^\tau \frac{dN(t) }{n\pi} \, \right] \\
&= o_{P \mid \mathcal{D}_n}(1) \cdot O_{P \mid \mathcal{D}_n}(1) \\
&= o_{P \mid \mathcal{D}_n}(1),
\end{aligned}
\end{equation}
and
\begin{equation}
\begin{aligned}
\sup_{\boldsymbol{\beta}} \left\| III \right\|
&\leq \sup_{t,\boldsymbol{\beta}} \left\| \overline{\mathbf{X}}^*(t,\boldsymbol{\beta}) - \overline{\mathbf{X}}(t,\boldsymbol{\beta}) \right\| \cdot \mathbf{P}_{\mathcal{D}_n}
\left[ \int_0^\tau \frac{dN(t)}{n\pi} \right] \\
&= o_{P \mid \mathcal{D}_n}(1) \cdot O(1) \\
&= o_{P \mid \mathcal{D}_n}(1).
\end{aligned}
\end{equation}
Hence,
\begin{equation}
\sup_{\boldsymbol{\beta} \in \Theta}
\left\| U^*(\boldsymbol{\beta}) - U(\boldsymbol{\beta}) \right\|
\leq \sup_{\boldsymbol{\beta}} \left( \left\| I \right\| + \left\| II \right\| + \left\| III \right\| \right)
= o_{P \mid \mathcal{D}_n}(1).
\end{equation}
This completes the proof of Lemma \ref{lem1}.
\end{proof}
\bmhead{Proof of Theorem \ref{th1}}
Note that the full-sample estimator $\hat{\boldsymbol{\beta}}_{full}$ is defined as the unique root of $U(\boldsymbol{\beta}) = 0,$ and since the parameter space $\Theta$ is compact and $U(\boldsymbol{\beta})$ is continuous on $\Theta$, it follows that $\hat{\boldsymbol{\beta}}_{full}$ is a well-separated zero of $U(\boldsymbol{\beta})$. By the general consistency theory for $Z$-estimators \citep{vaartvanderAsymptoticStatistics1998} and Lemma \ref{lem1}, the root of $U^*(\boldsymbol{\beta}) = 0$ converges in conditional probability to $\hat{\boldsymbol{\beta}}_{full}$. Therefore,
\begin{equation}
\hat{\boldsymbol{\beta}}_{sub} \xrightarrow{P \mid \mathcal{D}_n} \hat{\boldsymbol{\beta}}_{full}.
\end{equation}
Applying a first-order Taylor expansion of $U^*(\boldsymbol{\beta})$ around $\hat{\boldsymbol{\beta}}_{full}$, we have
\begin{equation}
0=U^*(\hat{\boldsymbol{\beta}}_{sub})
= U^*(\hat{\boldsymbol{\beta}}_{full})
+ \left. \frac{\partial U^*(\boldsymbol{\beta})}{\partial \boldsymbol{\beta}} \right|_{\boldsymbol{\beta} = \hat{\boldsymbol{\beta}}_{full}}
(\hat{\boldsymbol{\beta}}_{sub} - \hat{\boldsymbol{\beta}}_{full})
+ o_{P \mid \mathcal{D}_n} \left( 1 \right).
\end{equation}
It follows that
\begin{equation}\label{S2}
\sqrt{r}\left( \hat{\boldsymbol{\beta}}_{sub} - \hat{\boldsymbol{\beta}}_{full} \right)
= - \left[ \frac{\partial U^*(\hat{\boldsymbol{\beta}}_{full})}{\partial \boldsymbol{\beta}} \right]^{-1}
\sqrt{r}\, U^*(\hat{\boldsymbol{\beta}}_{full}).
\end{equation}
To proceed, observe that
\begin{equation}\label{S3}
\begin{aligned}
\sqrt{r}\, U^*(\hat{\boldsymbol{\beta}}_{full})
&= \sqrt{r}\, \mathbb{P}_{r \mid \mathcal{D}_n} \, \psi^*_{\hat{\boldsymbol{\beta}}_{full}}-\mathbf{P}_{\mathcal{D}_n}\psi_{\hat{\boldsymbol{\beta}}_{full}} \\
&= \sqrt{r}\left( \mathbb{P}_{r \mid \mathcal{D}_n} - \mathbf{P}_{\mathcal{D}_n} \right)
\psi_{\hat{\boldsymbol{\beta}}_{full}}
+ o_{P \mid \mathcal{D}_n}(r^{-1/2}),
\end{aligned}
\end{equation}
where the first equality holds since $\mathbf{P}_{\mathcal{D}_n}\psi_{\hat{\boldsymbol{\beta}}_{full}}
= U(\hat{\boldsymbol{\beta}}_{full}) = 0,$ and the second equality is given by  \eqref{S1}.

Consequently, from \eqref{S3}, we further obtain
\begin{equation}
\begin{aligned}
U^{*}(\hat{\boldsymbol{\beta}}_{full})
&= \left( \mathbb{P}_{r \mid \mathcal{D}_n} - \mathbf{P}_{\mathcal{D}_n} \right)
\psi_{\hat{\boldsymbol{\beta}}_{full}}
+ o_{P \mid \mathcal{D}_n}(r^{-1/2}) \\
&= \mathbb{P}_{r \mid \mathcal{D}_n} \psi_{\hat{\boldsymbol{\beta}}_{full}}
+ o_{P \mid \mathcal{D}_n}(r^{-1/2}) \\
&= \frac{1}{nr} \sum_{i=1}^{r} \frac{1}{\pi_i^*}
\int_{0}^{\tau} \left\{ \mathbf{X}_i^* - \overline{\mathbf{X}}_w(t,\hat{\boldsymbol{\beta}}_{full}) \right\}
\, dN_i^*(t)
+ o_{P \mid \mathcal{D}_n}(r^{-1/2}).
\end{aligned}
\end{equation}
We first consider the derivative of $U^*(\boldsymbol{\beta})$.
\begin{equation}\label{A27}
\frac{\partial U^*(\boldsymbol{\beta})}{\partial \boldsymbol{\beta}}
= \frac{1}{nr} \sum_{i=1}^r \frac{1}{\pi_i^*}
\int_0^\tau \left[
\frac{S^{(2)}(t,\boldsymbol{\beta})}{S^{(0)}(t,\boldsymbol{\beta})}
- \left( \frac{S^{(1)}(t,\boldsymbol{\beta})}{S^{(0)}(t,\boldsymbol{\beta})} \right)^{\otimes 2}
\right] dN_i^*(t)
+ o_{P \mid \mathcal{D}_n}(r^{-1/2}).
\end{equation}
Conditioning on $\mathcal{D}_n$, the expectation of the first term on the right-hand side of \eqref{A27} is given by
\begin{equation}
\begin{aligned}
&E\left[
\frac{1}{nr} \sum_{i=1}^r \frac{1}{\pi_i^*}
\int_0^\tau \left(
\frac{S^{(2)}(t,\boldsymbol{\beta})}{S^{(0)}(t,\boldsymbol{\beta})}
- \left( \frac{S^{(1)}(t,\boldsymbol{\beta})}{S^{(0)}(t,\boldsymbol{\beta})} \right)^{\otimes 2}
\right) dN_i^*(t)
\,\Big|\, \mathcal{D}_n
\right] \\
&\quad =
\frac{1}{n} \sum_{i=1}^n \int_0^\tau \left(
\frac{S^{(2)}(t,\boldsymbol{\beta})}{S^{(0)}(t,\boldsymbol{\beta})}
- \left( \frac{S^{(1)}(t,\boldsymbol{\beta})}{S^{(0)}(t,\boldsymbol{\beta})} \right)^{\otimes 2}
\right) dN_i(t).
\end{aligned}
\end{equation}
Let $U^*_{j_1 j_2}(\boldsymbol{\beta})$ denote the $(j_1,j_2)$-th component of the above quantity. Under Assumptions \ref{ass3}--\ref{ass4}, we have
\begin{equation}
\begin{aligned}
&E\Bigg[
U^*_{j_1 j_2}(\boldsymbol{\beta})
- \frac{1}{n} \sum_{i=1}^n \int_0^\tau
\left(
\frac{S^{(2)}(t,\boldsymbol{\beta})}{S^{(0)}(t,\boldsymbol{\beta})}
- \left(\frac{S^{(1)}(t,\boldsymbol{\beta})}{S^{(0)}(t,\boldsymbol{\beta})}\right)^{\otimes 2}
\right)_{j_1 j_2}
\, dN_i(t)
\;\Big|\; \mathcal{D}_n
\Bigg]^2 \\
&\leq \frac{1}{r n^2} \sum_{i=1}^n \frac{1}{\pi_i}
\left[\int_0^\tau\left(\frac{S^{(2)}(t,\boldsymbol{\beta})}{S^{(0)}(t,\boldsymbol{\beta})}
- \left(\frac{S^{(1)}(t,\boldsymbol{\beta})}{S^{(0)}(t,\boldsymbol{\beta})}\right)^{\otimes 2}
\right)_{j_1 j_2}
\, dN_i(t)
\right]^2 \\
&= O_{P \mid \mathcal{D}_n}(r^{-1}).
\end{aligned}
\end{equation}
By Markov's inequality, it follows that
\begin{equation}
\frac{\partial U^*(\boldsymbol{\beta})}{\partial \boldsymbol{\beta}}
= \frac{\partial U(\boldsymbol{\beta})}{\partial \boldsymbol{\beta}}
+ O_{P \mid \mathcal{D}_n}(r^{-1/2}).
\end{equation}
Evaluating at $\hat{\boldsymbol{\beta}}_{full}$, we obtain
\begin{equation}\label{S4}
\begin{aligned}
\frac{\partial U^*(\hat{\boldsymbol{\beta}}_{full})}{\partial \boldsymbol{\beta}}
&= \frac{1}{n} \sum_{i=1}^n \int_0^\tau
\left[
\frac{S^{(2)}(t,\hat{\boldsymbol{\beta}}_{full})}{S^{(0)}(t,\hat{\boldsymbol{\beta}}_{full})}
- \left(\frac{S^{(1)}(t,\hat{\boldsymbol{\beta}}_{full})}{S^{(0)}(t,\hat{\boldsymbol{\beta}}_{full})}\right)^{\otimes 2}
\right] dN_i(t)
+ O_{P \mid \mathcal{D}_n}(r^{-1/2}) \\
&= \boldsymbol{\Psi} + O_{P \mid \mathcal{D}_n}(r^{-1/2}).
\end{aligned}
\end{equation}
Finally, since $\mathcal{F}_n$ is a Donsker class, the central limit theorem yields
\begin{equation}\label{S5}
\sqrt{r}\left( \mathbb{P}_{r \mid \mathcal{D}_n} - \mathbf{P}_{\mathcal{D}_n} \right)
\psi_{\hat{\boldsymbol{\beta}}_{full}}
\xrightarrow{d}
N(0,\boldsymbol{\Gamma}),
\end{equation}
where the variance-covariance matrix $\boldsymbol{\Gamma}$ is defined as
\begin{equation}
\begin{aligned}
\boldsymbol{\Gamma}
&= \mathrm{Var}_{\mathbf{P}_{\mathcal{D}_n}} \left( \psi_{\hat{\boldsymbol{\beta}}_{full}} \right) \\
&= E_{\mathbf{P}_{\mathcal{D}_n}} \left[
\left( \psi_{\hat{\boldsymbol{\beta}}_{full}} - E_{\mathbf{P}_{\mathcal{D}_n}}(\psi_{\hat{\boldsymbol{\beta}}_{full}}) \right)
\left( \psi_{\hat{\boldsymbol{\beta}}_{full}} - E_{\mathbf{P}_{\mathcal{D}_n}}(\psi_{\hat{\boldsymbol{\beta}}_{full}}) \right)^\top
\right] \\
&= E_{\mathbf{P}_{\mathcal{D}_n}} \left[
\psi_{\hat{\boldsymbol{\beta}}_{full}} \psi_{\hat{\boldsymbol{\beta}}_{full}}^\top
\right] \\
&= \frac{1}{n^2} \sum_{i=1}^n \frac{1}{\pi_i}
\left[
\int_0^\tau \left\{ \mathbf{X}_i - \overline{\mathbf{X}}_w(t,\hat{\boldsymbol{\beta}}_{full}) \right\}
\, dN_i(t)
\right]^{\otimes 2}.
\end{aligned}
\end{equation}
Substituting \eqref{S4} and \eqref{S5} into \eqref{S2} and applying Slutsky's theorem, we obtain
\begin{equation}
\sqrt{r}\left( \hat{\boldsymbol{\beta}}_{sub} - \hat{\boldsymbol{\beta}}_{full} \right)
\xrightarrow{d}
N(0,\boldsymbol{\Sigma}),
\end{equation}
where
\begin{equation}
\boldsymbol{\Sigma} = \boldsymbol{\Psi}^{-1} \boldsymbol{\Gamma} \boldsymbol{\Psi}^{-1}.
\end{equation}

\bmhead{Proof of Theorem \ref{th2}} 
Based on Theorem \ref{th1}, the asymptotic covariance matrix is $\boldsymbol{\Sigma} = \boldsymbol{\Psi}^{-1} \boldsymbol{\Gamma} \boldsymbol{\Psi}^{-1}$, then the trace of $\boldsymbol{\Sigma}$ can be explicitly expanded as
$$
\begin{aligned}
\mathrm{tr}(\boldsymbol{\Sigma})
&= \mathrm{tr}\Bigg(
\frac{1}{n^{2}r} \sum_{i=1}^{n} \frac{1}{\pi_{i}}
\Psi^{-1}
\Big[\int_{0}^{\tau} (\mathbf{X}_{i} - \overline{\mathbf{X}}_{w}(t,\hat{\boldsymbol{\beta}}_{full})) dN_{i}(t) \Big]^{\otimes 2}
(\Psi^{-1})^\top
\Bigg) \\
&= \frac{1}{r n^{2}} \sum_{i=1}^{n} \frac{1}{\pi_{i}}
\Big\| \Psi^{-1} \int_{0}^{\tau} (\mathbf{X}_{i} - \overline{\mathbf{X}}_{w}(t,\hat{\boldsymbol{\beta}}_{full})) dN_{i}(t) \Big\|^{2} \\
&= \frac{1}{r n^2} \sum_{i=1}^n \pi_i \sum_{i=1}^n \frac{1}{\pi_i}
\Big\| \Psi^{-1} \int_0^\tau (\mathbf{X}_i - \overline{\mathbf{X}}_w(t,\hat{\boldsymbol{\beta}}_{full})) dN_i(t) \Big\|^{2} \\
&\ge \frac{1}{r n^{2}}
\Bigg\{ \sum_{i=1}^{n} \Big\| \Psi^{-1} \int_{0}^{\tau} (\mathbf{X}_{i} - \overline{\mathbf{X}}_{w}(t,\hat{\boldsymbol{\beta}}_{full})) dN_{i}(t) \Big\| \Bigg\}^{2},
\end{aligned}
$$
where the inequality follows from the Cauchy-Schwarz inequality under the constraint that $\sum_{i=1}^{n} \pi_i = 1$. The lower bound is attained if and only if the subsampling probabilities $\pi_i$ satisfy
$$\pi_i = \frac{\left\| \boldsymbol{\Psi}^{-1} \int_{0}^{\tau} (\mathbf{X}_{i} - \overline{\mathbf{X}}_{w}(t,\hat{\boldsymbol{\beta}}_{full})) dN_{i}(t) \right\|}{\sum_{j=1}^n \left\| \boldsymbol{\Psi}^{-1} \int_{0}^{\tau} (\mathbf{X}_{j} - \overline{\mathbf{X}}_{w}(t,\hat{\boldsymbol{\beta}}_{full})) dN_{j}(t) \right\|}, \quad i=1, \dots, n.$$

\bmhead{Proof of Theorem \ref{th3}}  
Similar to the derivation in Theorem \ref{th2}, we have
$$
\begin{aligned}
\mathrm{tr}(\boldsymbol{\Gamma})
&= \mathrm{tr}\Bigg(
\frac{1}{n^{2}r} \sum_{i=1}^{n} \frac{1}{\pi_{i}}
\Big[\int_{0}^{\tau} (\mathbf{X}_{i} - \overline{\mathbf{X}}_{w}(t,\hat{\boldsymbol{\beta}}_{full})) dN_{i}(t) \Big]^{\otimes 2} \Bigg) \\
&= \frac{1}{r n^{2}} \sum_{i=1}^{n} \frac{1}{\pi_{i}}
\Big\|\int_{0}^{\tau} (\mathbf{X}_{i} - \overline{\mathbf{X}}_{w}(t,\hat{\boldsymbol{\beta}}_{full})) dN_{i}(t) \Big\|^{2} \\
&= \frac{1}{r n^2} \sum_{i=1}^n \pi_i \sum_{i=1}^n \frac{1}{\pi_i}
\Big\|\int_0^\tau (\mathbf{X}_i - \overline{\mathbf{X}}_w(t,\hat{\boldsymbol{\beta}}_{full})) dN_i(t) \Big\|^{2} \\
&\ge \frac{1}{r n^{2}}
\Bigg\{ \sum_{i=1}^{n} \Big\|\int_{0}^{\tau} (\mathbf{X}_{i} - \overline{\mathbf{X}}_{w}(t,\hat{\boldsymbol{\beta}}_{full})) dN_{i}(t) \Big\| \Bigg\}^{2}.
\end{aligned}
$$
The lower bound is attained if and only if the subsampling probabilities $\pi_i$ satisfy
$$\pi_i = \frac{\left\| \int_{0}^{\tau} (\mathbf{X}_{i} - \overline{\mathbf{X}}_{w}(t,\hat{\boldsymbol{\beta}}_{full})) dN_{i}(t) \right\|}{\sum_{j=1}^n \left\| \int_{0}^{\tau} (\mathbf{X}_{j} - \overline{\mathbf{X}}_{w}(t,\hat{\boldsymbol{\beta}}_{full})) dN_{j}(t) \right\|}, \quad i=1, \dots, n.$$

Finally, to prove Theorem \ref{th4}, it is necessary to show Lemma \ref{lem2} in the following.
\begin{lemma}\label{lem2}\normalsize
     Under Assumptions \ref{ass1}-\ref{ass3}, if
$$
\sup_{t\in[0,\tau],\,\boldsymbol{\beta}\in\Theta}
\left|
\bar{\mathbf{X}}_w^{*}(t,\boldsymbol{\beta})-\bar{\mathbf{X}}_w(t,\boldsymbol{\beta})
\right|
\xrightarrow{P|\mathcal D_n}0, \text{ and } 
\|\hat{\boldsymbol{\beta}}_{r_0}-\hat{\boldsymbol{\beta}}_{{full}}\|
=O_{P|\mathcal D_n}(r_0^{-1/2}),
$$
then for all $1\le i\le n$,
\begin{equation}
\max_{1\le i\le n}\left|
\frac{\pi_i^{{Lopt,rob}}(\lambda)-\pi_i^{{app,rob}}(\lambda)}
{\pi_i^{{app,rob}}(\lambda)}
\right|
=O_{P|\mathcal D_n,\hat{\boldsymbol{\beta}}_{r_0}}(r_0^{-1/2})
=o_{P|\mathcal D_n,\hat{\boldsymbol{\beta}}_{r_0}}(1),
\end{equation}
and
\begin{equation}
\max_{1\le i\le n}\left|
\frac{\pi_i^{{Aopt,rob}}(\lambda)-\pi_i^{{app,rob}}(\lambda)}
{\pi_i^{{app,rob}}(\lambda)}
\right|
=O_{P|\mathcal D_n,\hat{\boldsymbol{\beta}}_{r_0}}(r_0^{-1/2})
=o_{P|\mathcal D_n,\hat{\boldsymbol{\beta}}_{r_0}}(1).
\end{equation}
\end{lemma}
\begin{proof}[Proof.]
\normalsize   
We only prove the L-opt case, the A-opt case follows similarly.
For notational convenience, define
\[
\pi_i^{{Lopt}}
=\frac{A_i(\hat{\boldsymbol{\beta}}_{{full}})}{\sum_{j=1}^nA_j(\hat{\boldsymbol{\beta}}_{{full}})},\text{ and } 
\pi_i^{{app}}
=\frac{A_i^*(\hat{\boldsymbol{\beta}}_{r_0})}{\sum_{j=1}^nA_j^*(\hat{\boldsymbol{\beta}}_{r_0})},
\]
where
\[
A_i(\boldsymbol{\beta})=\Big\|\int_0^\tau\{\mathbf{X}_i-\bar{\mathbf{X}}_w(t,\boldsymbol{\beta})\}\,dN_i(t)\Big\|,\text{ and } 
A_i^*(\boldsymbol{\beta})=\Big\|\int_0^\tau\{\mathbf{X}_i-\bar{\mathbf{X}}_w^*(t,\boldsymbol{\beta})\}\,dN_i(t)\Big\|.
\]
By the triangle inequality,
\begin{equation}
\begin{aligned}
\sup_i |A_i(\hat{\boldsymbol{\beta}}_{{full}})-A_i^*(\hat{\boldsymbol{\beta}}_{r_0})|
&\le
\sup_i |A_i(\hat{\boldsymbol{\beta}}_{{full}})-A_i^*(\hat{\boldsymbol{\beta}}_{{full}})|
+\sup_i |A_i^*(\hat{\boldsymbol{\beta}}_{{full}})-A_i^*(\hat{\boldsymbol{\beta}}_{r_0})| \\
&=o_{P|\mathcal D_n}(1)
+O_{P|\mathcal D_n,\hat{\boldsymbol{\beta}}_{r_0}}
(\|\hat{\boldsymbol{\beta}}_{r_0}-\hat{\boldsymbol{\beta}}_{{full}}\|) \\
&=O_{P|\mathcal D_n,\hat{\boldsymbol{\beta}}_{r_0}}(r_0^{-1/2}).
\end{aligned}
\end{equation}
Let $B=\sum_{j=1}^nA_j(\hat{\boldsymbol{\beta}}_{{full}})$ and
$B^*=\sum_{j=1}^nA_j^*(\hat{\boldsymbol{\beta}}_{r_0})$. Then
\begin{equation}
|B-B^*|\le\sum_{j=1}^n|A_j(\hat{\boldsymbol{\beta}}_{{full}})-A_j^*(\hat{\boldsymbol{\beta}}_{r_0})|
=O_{P|\mathcal D_n,\hat{\boldsymbol{\beta}}_{r_0}}(nr_0^{-1/2}).
\end{equation}
Thus,
\begin{equation}
|\pi_i^{{Lopt}}-\pi_i^{{app}}|\le\frac{|A_i-A_i^*|}{B^*}+\frac{A_i|B-B^*|}{BB^*}.
\end{equation}
Since $B=O_P(n)$, $B^*=O_P(n)$, and $A_i=O_P(1)$, we obtain
\begin{equation}
\max_i|\pi_i^{{Lopt}}-\pi_i^{{app}}|
=O_{P|\mathcal D_n,\hat{\boldsymbol{\beta}}_{r_0}}(n^{-1}r_0^{-1/2}).
\end{equation}
Using $\pi_i^{{rob}}\ge \lambda/n$, we have
\begin{equation}
\begin{aligned}
\max_i\left|\frac{\pi_i^{{Lopt,rob}}-\pi_i^{{app,rob}}}
{\pi_i^{{app,rob}}}\right|
&\le\frac{n}{\lambda}\max_i|\pi_i^{{Lopt}}-\pi_i^{{app}}| \\
&=O_{P|\mathcal D_n,\hat{\boldsymbol{\beta}}_{r_0}}(r_0^{-1/2}).
\end{aligned}
\end{equation}
\end{proof}

\bmhead{Proof of Theorem \ref{th4}} 
According to the standard linear representation of M-estimators, it suffices to show that the feasible score function
$U_r^*(\boldsymbol{\beta};\pi^{{app,rob}})$ is asymptotically equivalent to the theoretical score function
$U_r^*(\boldsymbol{\beta};\pi^{{Lopt,rob}})$ at $\hat{\boldsymbol{\beta}}_{{full}}$. Define
\begin{equation}
\Delta_r
=
\left[
U_r^*(\hat{\boldsymbol{\beta}}_{{full}};\pi^{{app,rob}})
-
U_r^*(\hat{\boldsymbol{\beta}}_{{full}};\pi^{{Lopt,rob}})
\right].
\end{equation}
Let
\begin{equation}
\phi_i^*=\int_0^\tau\{\mathbf{X}_i-\bar{\mathbf{X}}_w^*(t,\hat{\boldsymbol{\beta}}_{{full}})\}\,dN_i(t),
\end{equation}
then
\begin{equation}
\begin{aligned}\label{A45}
\|\Delta_r\| 
&= \frac{1}{nr} \sum_{i=1}^r
\left( \frac{\pi_i^{*{Lopt,rob}} - \pi_i^{*{app,rob}}}{\pi_i^{*{app,rob}}} \right)
\frac{\phi_i^*}{\pi_i^{*{Lopt,rob}}} \\
&\le O_{P|\mathcal{D}_n}(1) \cdot
O_{P|\mathcal{D}_n, \hat{\boldsymbol{\beta}}_{r_0}}(r_0^{-1/2}) \\
&= o_{P|\mathcal{D}_n, \hat{\boldsymbol{\beta}}_{r_0}}(1),
\end{aligned}
\end{equation}
The inequality above is obtained by Lemma \ref{lem2}.

Moreover,
\begin{equation}
\begin{aligned}\label{A46}
\|\tilde{\Gamma}-\Gamma\|&=\left\|\frac{1}{n^2}\sum_{i=1}^n\left(\frac{1}{\pi_i^{{app,rob}}}-\frac{1}{\pi_i^{{Lopt,rob}}}\right)\phi_i^{\otimes2}
\right\| \\&=\left\|\frac{1}{n^2}\sum_{i=1}^n\left(\frac{\pi_i^{{Lopt,rob}}-\pi_i^{{app,rob}}}{\pi_i^{{app,rob}}}\right)\frac{\phi_i^{\otimes2}}
{\pi_i^{{Lopt,rob}}}\right\| \\&\le\max_i\left|\frac{\pi_i^{{Lopt,rob}}-\pi_i^{{app,rob}}}{\pi_i^{{app,rob}}}\right|\cdot\left\|\frac{1}{n^2}
\sum_{i=1}^n\frac{\phi_i^{\otimes2}}{\pi_i^{{Lopt,rob}}}\right\| \\
&=o_{P|\mathcal D_n,\hat{\boldsymbol{\beta}}_{r_0}}(1).
\end{aligned}
\end{equation}
Combining (\ref{A45}) and (\ref{A46}), the proof is complete.
\bmhead{Proof of Corollary \ref{cor1}} 
We first establish the unconditional consistency and convergence rate of the subsample estimator $\tilde{\boldsymbol{\beta}}_{sub}$. By the triangle inequality,
\begin{equation}
\|\tilde{\boldsymbol{\beta}}_{sub} - \boldsymbol{\beta}_0\|
\le \|\tilde{\boldsymbol{\beta}}_{sub} - \hat{\boldsymbol{\beta}}_{full}\|
+ \|\hat{\boldsymbol{\beta}}_{full} - \boldsymbol{\beta}_0\|.
\end{equation}

Combining the conditional error of the subsampling estimator,
$O_{P \mid \mathcal{D}_n, \hat{\boldsymbol{\beta}}_{r_0}}(r^{-1/2})$, with the full-sample error $O_P(n^{-1/2})$, we obtain
\begin{equation}
\|\tilde{\boldsymbol{\beta}}_{sub} - \boldsymbol{\beta}_0\|
= O_{P \mid \mathcal{D}_n, \hat{\boldsymbol{\beta}}_{r_0}}(r^{-1/2})
+ O_P(n^{-1/2}).
\end{equation}
Since $r = o(n)$, it follows from Proposition 2 of \cite{wang2022sampling}
 that
\begin{equation}
\|\tilde{\boldsymbol{\beta}}_{sub} - \boldsymbol{\beta}_0\| = O_P(r^{-1/2}).
\end{equation}
Furthermore,
\begin{equation}
\begin{aligned}
\sqrt{r}(\tilde{\boldsymbol{\beta}}_{sub} - \boldsymbol{\beta}_0)
&= \sqrt{r}(\tilde{\boldsymbol{\beta}}_{sub} - \hat{\boldsymbol{\beta}}_{full})
+ \sqrt{r}(\hat{\boldsymbol{\beta}}_{full} - \boldsymbol{\beta}_0) \\
&= \sqrt{r}(\tilde{\boldsymbol{\beta}}_{sub} - \hat{\boldsymbol{\beta}}_{full})
+ o_P(1).
\end{aligned}
\end{equation}
Therefore, conditional on $\mathcal{D}_n$ and $\hat{\boldsymbol{\beta}}_{r_0}$, the asymptotic distribution of
$\sqrt{r}(\tilde{\boldsymbol{\beta}}_{sub} - \boldsymbol{\beta}_0)$ coincides with that of
$\sqrt{r}(\tilde{\boldsymbol{\beta}}_{sub} - \hat{\boldsymbol{\beta}}_{full})$.

By Proposition 2 of \cite{wang2022sampling}, we have
\begin{equation}
\sqrt{r}(\tilde{\boldsymbol{\beta}}_{sub} - \boldsymbol{\beta}_0)
\xrightarrow{d} N(0,\tilde{\boldsymbol{\Sigma}}).
\end{equation}
This completes the proof of Corollary \ref{cor1}.



\bibliographystyle{apalike}
\bibliography{interactapasample}

\end{document}